\documentclass[journal,onecolumn]{IEEEtran}
\usepackage{url}
\usepackage{color}
\usepackage{verbatim}
\usepackage{amsmath}
\usepackage{amssymb}
\usepackage{mathrsfs}
\usepackage{makecell}
\usepackage{threeparttable}
\usepackage[figuresright]{rotating}
\usepackage{soul} 
\usepackage{color, xcolor} 
\usepackage{booktabs}
\usepackage{bm}
\usepackage{multicol}
\usepackage{xcolor}
\usepackage{cite}

\usepackage{setspace} 
\makeatletter
\let\NAT@parse\undefined
\makeatother
\usepackage{hyperref}  
\usepackage{diagbox}

\hypersetup{
	colorlinks=true,
	linkcolor=black,
	citecolor=black
}

\usepackage{tikz}
\usetikzlibrary{shapes,arrows}

\newtheorem{lemma}{\textbf{Lemma}}

\newtheorem{theorem}{\textbf{Theorem}}

\newtheorem{proposition}{\textbf{Proposition}}
\newtheorem{definition}{\textbf{Definition}}
\newenvironment{proof}{{\noindent\it \textbf{Proof}:} }{\hfill $\square$\par}
\newtheorem{example}{\textbf{Example}}

\newtheorem{remark}{\textbf{Remark}}

\begin{document}
	
	\title{Lower Bounds for Error Coefficients of Griesmer Optimal Linear Codes via Iteration
	}

	\author{Chaofeng Guan, Shitao Li, Gaojun Luo, Zhi Ma, Hong Wang ~\IEEEmembership{}
	}


	\maketitle

	\begin{abstract}
		The error coefficient of a linear code is defined as the number of minimum-weight codewords. In an additive white Gaussian noise channel, optimal linear codes with the smallest error coefficients achieve the best possible asymptotic frame error rate (AFER) among all optimal linear codes under maximum likelihood decoding. Such codes are referred to as AFER-optimal linear codes.
		
		The Griesmer bound is essential for determining the optimality of linear codes. However, establishing tight lower bounds on the error coefficients of Griesmer optimal linear codes is challenging, and the linear programming bound often performs inadequately. In this paper, we propose several iterative lower bounds for the error coefficients of Griesmer optimal linear codes. Specifically, for binary linear codes, our bounds are tight in most cases when the dimension does not exceed $5$. To evaluate the performance of our bounds when they are not tight, we also determine the parameters of the remaining 5-dimensional AFER-optimal linear codes. Our final comparison demonstrates that even when our bounds are not tight, they remain very close to the actual values, with a gap of less than or equal to $2$.

		%
		%
		%
	\end{abstract}
	
	
	\IEEEpeerreviewmaketitle

	\section{Introduction}
	
	\IEEEPARstart{L}{et} $q$ be a prime power and $\mathbb{F}_q$ denote the finite field of $q$ elements. An $[n,k]_q$ linear code $C$ is a $k$-dimensional linear subspace of $\mathbb{F}_q^n$. The vectors in $C$ are called codewords. 
	The Hamming weight of a codeword $\mathbf{c} \in C$ is defined as the number of nonzero coordinates in $\mathbf{c}$. 
	The minimum Hamming weight of a linear code $C$ is the smallest Hamming weight among all nonzero codewords in $C$. 	
	A linear code $C$ with minimum Hamming weight $d$ is denoted as an $[n,k,d]_q$ linear code\footnote{For a linear code, its minimum Hamming weight is equal to its minimum Hamming distance. }. An $[n,k,d]_q$ linear code encodes $k$ information bits into an $n$-bit codeword. This code can correct up to $\left\lfloor (d-1)/2 \right\rfloor$ bit errors. Given $n$ and $k$, a larger $d$ provides stronger error-correcting capabilities. When $n$ and $d$ are fixed, a larger $k$ results in a higher code rate $k/n$, thereby improving system efficiency.
	Due to the interdependent relationships among the three parameters $n$, $k$, and $d$, establishing bounds for linear codes has always been a central issue in classical coding theory. The well-known bounds include the Singleton bound \cite{Singleton1964}, the Griesmer bound \cite{griesmer1960bound}, the Plotkin bound \cite{plotkin1960binary}, and the linear programming bound \cite{delsarte1972bounds}. These bounds have guided the construction of many optimal linear codes that achieve these limits \cite{HengZiling2020,ShiMinjiaSO2023,li2023hull,HuZhao2024}.

	In practice, linear codes with the same parameters $n$, $k$, and $d$ may have different frame error rates (FER) denoted by $P_e$. According to \cite{swaszek1995lower}, for the additive white Gaussian noise (AWGN) channel, the FER $P_e$ of an $[n,k]_2$ linear code $C$ under maximum likelihood (ML) decoding satisfies 
    \begin{equation} P_{e} \simeq A_{d}(C) Q\left(\sqrt{ \frac{2 dkE_{b}}{nN_{0} } }\right), \quad \frac{E_{b}}{N_{0}} \rightarrow \infty, 
    \end{equation} 
    where $\simeq$ denotes asymptotic equality, $Q(x) \triangleq \frac{1}{\sqrt{2 \pi}} \int_{x}^{\infty} e^{-\frac{t^{2}}{2}} d t$, $E_{b} / N_{0}$ is the signal-to-noise ratio (SNR), and $A_d(C)$ is the number of minimum-weight codewords in $C$. If a code $C$ has the largest minimum distance $d$ and the smallest error coefficient $A_d(C)$ among all $[n,k]_q$ linear codes, then $C$ achieves the lowest \textit{asymptotic frame error rate (AFER)} among all optimal linear codes under ML decoding. Therefore, the value of $A_d(C)$ is also an important parameter for linear codes. An $[n,k,d]_q$ linear code with the smallest $A_d(C)$ is called \textit{AFER-optimal} \cite{abdullah2023some}.
	
	In modern coding theory, Ar{\i}kan discovered polar codes \cite{arikan2009channel} in 2009. Due to their excellent performance, polar codes were selected as the coding scheme for the 5G control channel. Subsequently, Ar{\i}kan proposed polarization-adjusted convolutional (PAC) codes \cite{arikan2019sequential}, a variant of polar codes. He demonstrated that PAC codes perform remarkably close to the Polyanskiy-Poor-Verd{\'u} bound \cite{polyanskiy2010channel} over the binary input additive white Gaussian noise (AWGN) channel in the low-to-moderate signal-to-noise ratio (SNR) regime. The minimum distances and error coefficients also affect the performance of polar and PAC codes. Many efforts have been made to optimize these parameters to enhance their performance \cite{Rowshan2023, Rowshan2023a, Gu2024, Dragoi2024, moradi2024polarization}. Classical and modern coding theories are interconnected. For example, polar codes are subcodes of Reed-Muller codes \cite{huffman2010fundamentals}. Lin et al. \cite{lin2020transformation} demonstrated that binary linear codes can be converted into polar codes. As a result, binary linear codes can also be decoded using polar decoding techniques \cite{khebbou2023decoding,khebbou2023single}.

    The MacWilliams formula \cite{macwilliams1977theory} uses the weight enumerator to determine bounds on the error coefficients of linear codes through linear programming or polynomial methods. 
    In 2021, Sol{\'e} et al. \cite{sole2021linear} proposed a new linear programming bound for the lower and upper limits of the error coefficients of linear codes. Later in \cite{abdullah2023some, abdullah2023new},  Abdullah and Mow constructed many Griesmer optimal binary linear and PAC codes that achieve Sol{\'e}'s lower bound. However, the linear programming bound in \cite{sole2021linear} often fails to provide tight bounds in many scenarios \cite{abdullah2023some, abdullah2023new}. Moreover, it requires complex calculations and can only compute linear codes for specific parameters. 
    In \cite{ShitaoLi2024}, Li et al. established two bounds on the error coefficients of linear codes through shortening and extension, constructed several classes of AFER-optimal binary codes, and resolved conjectures in \cite{abdullah2023some}. Using the classification in \cite{li2024characterization}, they also determined the parameters of all binary AFER-optimal linear codes up to length $13$.
    For Griesmer optimal cases, the lower bounds on the error coefficients of linear codes are still not well determined. \textbf{Therefore, more powerful bounds are needed to determine the AFER-optimal property of linear codes}. This motivates us to investigate bounds on linear codes regarding $n$, $k$, $d$, and $A_d(C)$.

	\textbf{Our Contributions}
	
	This paper proposes several practical iterative bounds for Griesmer optimal linear codes.
	We summarize the contributions as follows.
	
	\begin{enumerate}
		\item
		Determining the lower bounds on the error coefficients of optimal linear codes is well known to be very difficult. In \cite{liu2023kissing}, the authors reviewed the currently known methods for computing lower bounds on the error coefficients of linear codes. However, linear programming and polynomial methods tend to perform well only in special cases and often do not meet expectations. For example, Abdullah and Mow \cite{abdullah2023new, abdullah2023some} constructed many binary Griesmer optimal linear codes with small error coefficients. However, their AFER-optimality cannot be determined using the linear programming method. With our bounds, their AFER-optimality can be easily ascertained.

		Specifically, we propose five lower bounds on the error coefficients of Griesmer optimal linear codes. These bounds are derived by analyzing the spatial relationships between linear codes and their related residual codes, as well as the structure of the codeword space. Together, these five bounds form a strong lower bound for the error coefficients of Griesmer optimal linear codes. Under different constraints, we formulate these five lower bounds in Equation (\ref{Eq_five_lower_bounds}). Starting with binary 2-dimensional AFER-optimal linear codes, we iteratively provide lower bounds on the error coefficients for Griesmer optimal linear codes with dimensions 3, 4, and 5. Numerical results indicate that our bounds are tight in most cases, as shown in Tables \ref{Three_Binary_AFER}, \ref{Four_Binary_AFER}, and \ref{Five_Binary_AFER}.

		
		%
		\item
		Using our proposed bounds and the classification results of special binary Griesmer codes from \cite{helleseth1984further}, we determine the parameters of all binary AFER-optimal linear codes with dimensions up to $5$. Additionally, we provide explicit constructions for these codes. To evaluate the performance of our bounds when they are not tight, we compare the gap between our bounds and the actual values, as shown in Table \ref{Five_Binary_AFER_2}. The results indicate that even when our bounds are not tight, they remain very close to the real values, with a gap of less than or equal to $2$ for $k=5$.
	\end{enumerate}

	\textbf{Organization}

	In Section \ref{Sec II}, we introduce some primary results.
	We propose five iterative bounds for AFER-optimal linear codes in Section \ref{Sec III}.
	In Section \ref{Sec V}, we determine the parameters of AFER-optimal linear codes when our bounds are not tight for $k=5$.
	Finally, Section \ref{Sec VI} concludes the paper.
	
	\section{Preliminaries}\label{Sec II}
	Throughout this paper, let $\mathbb{F}_q=\{0,1,\alpha,\ldots,\alpha^{q-2}\}$ denote the finite field with $q$ elements, and let $\mathbb{F}_q^*=\mathbb{F}_q\setminus \{0\}$.
	For $\alpha\in \mathbb{F}_q$, we use $\bm{\alpha}_{m\times n}$ (resp. $\bm{\alpha}_n$) to denote an $m\times n$ matrix with each entry being $\alpha$ (resp. vector of length $n$). If the sizes of matrices (or vectors) are clear, we will omit the subscript.
	Let $[n]$ denote the set $\{1,2,\ldots,n \}$.  
    For $m<n$, let $[m,n]$ denote the set $\{m+1,\ldots,n \}$.

	\subsection{Basics of linear codes}
	
	An $[n,k]_q$ linear code $C$ is a $k$-dimensional subspace of $\mathbb{F}_q^n$.
	For a vector $\mathbf{u}=(u_1,u_2,\ldots,u_{n})$ of $\mathbb{F}_q^n$, we respectively define $\text{Supp}(\mathbf{u})=\{i\mid  u_i\ne \mathbf{0} \}$ and $wt(\mathbf{u})=| \text{Supp}(\mathbf{u})|$ as its support and Hamming weight.
	The support of $C$ is $\text{Supp}(C)=\bigcup _{c \in C} \text{Supp}(\mathbf{c})$. Let $n(C)=|\text{Supp}(C)|$ be the \textit{effective length} of $C$.
	The minimum Hamming weight and error coefficient of $C$ are $d(C)=\min\{wt(\mathbf{c})\mid \mathbf{c}\in C \}$ and $e(C)=|\{ \mathbf{c}\mid  wt(\mathbf{c})=d(C),  \mathbf{c}\in C\}|$, respectively.
	If $C$ is an $[n,k]_q$ linear code of minimum Hamming weight $d$ and error coefficient $e$, then we denote the parameters of $C$ by $[n,k,d;e]_q$.

	Let $A_i(C)=|\{ \mathbf{c}\mid  wt(\mathbf{c})=i,  \mathbf{c}\in C\}|$, where $0\le i \le n$.
	For $1\le i\le n$, $(q-1)\mid A_i(C)$ or $A_i(C)=0$.
	The polynomial $\sum_{i=0}^n A_{i}(C)x^i$ is the weight enumerator of $C$.
	For a set of $\mathbb{F}_q$-linearly independent vectors $\mathbf{c}_1$, $\mathbf{c}_2$, $\ldots$, $\mathbf{c}_k$ in $\mathbb{F}_q^n$, let $\left \langle \mathbf{c}_1,\mathbf{c}_2,\ldots,\mathbf{c}_k \right \rangle$ denote the vector space generated by $\mathbf{c}_1,\mathbf{c}_2,\ldots,\mathbf{c}_k$.
	If $C=\left \langle \mathbf{c}_1,\mathbf{c}_2,\ldots,\mathbf{c}_{k^{\prime}} \right \rangle$, where $k^{\prime}<k$, then we denote $\left \langle C,\mathbf{c}_{k^{\prime}+1},\ldots,\mathbf{c}_k \right \rangle=\left \langle \mathbf{c}_1,\mathbf{c}_2,\ldots,\mathbf{c}_{k} \right \rangle$. 
    For $t$ linear codes $C_1$, $C_2$, $\ldots$, and $C_t$ of the same dimension, we denote by $C_1\times C_2\times\cdots\times C_t$ their juxtaposition code. 
	The extended code of an $[n,k]_q$ linear code $C$ is
	\begin{equation}
		\widehat{C} =\left\{\left(c_0,c_1,\ldots,c_{n-1},\sum_{i=0}^{n-1}c_i \right)\mid \left(c_0,c_1,\ldots,c_{n-1} \right)\in C \right \}.
	\end{equation}
	
	\begin{lemma}(\cite{huffman2010fundamentals}, p. 15) \label{Lem_binary-odd-even}
		If $C$ is an $[n,k,d]_2$ linear code with $2\nmid d$, then the extended code of $C$ is an $[n+1,k,d+1]_2$ linear code.
	\end{lemma}

	\subsection{Bounds on linear codes}

	The Griesmer bound is crucial for determining the optimality of linear codes.
	
	\begin{theorem}
		(Griesmer Bound, \cite{griesmer1960bound})  \label{Griesmer_Bound}
		If $\mathcal{C}$ is an $[n,k,d]_{q}$ linear code, then
		\begin{equation} 
			n \geq g_q(k, d) :\triangleq \sum\limits_{i=0}^{k-1}\left\lceil  \frac{d}{q^{i}}\right\rceil.
		\end{equation}
	\end{theorem}
	
	The linear code reaching this bound with equality is called a \textit{Griesmer code}.
	
	\begin{theorem}
		(Modified Griesmer Bound, \cite{mceliece1991modifications})  \label{Modifications_Griesmer_Bound}
		If $\mathcal{C}$ is an $[n,k\ge 2,d]_{q}$ linear code with maximum weight $m$, then,
		\begin{equation} 
			n \geq \mathsf{g}_q(k, d,m) :\triangleq \sum\limits_{i=0}^{k-2}\left\lceil  \frac{d}{q^{i}}\right\rceil+\left\lceil  \frac{m}{q^{k-1}}\right\rceil.
		\end{equation}
	\end{theorem}

	Let $n(k,d,q)$ denote the smallest value of the length of an $[n,k]_q$ linear code with dimension $k$ and minimum Hamming weight $d$.
	Let $d(n,k,q)$ denote the largest value of the minimum distance of $[n, k]_q$ linear code.
	Let $e(n,k,q)$ denote the smallest error coefficient of $[n,k,d(n,k,q)]_q$ linear code.
	Let $e_d(n,k,q)$ denote the smallest value of $A_d(C)$ of $[n,k,d]_q$ linear code.
	It follows that $e_{d}(n,k,q)=e(n,k,q)$ for $d=d(n,k,q)$ and $e_{d}(n,k,q)=0$ for $d< d(n,k,q)$.

	If $C$ is an $[n,k,d(n,k,q)]_q$ linear code, then we call $C$ a \textit{distance-optimal} linear code.
	If $C$ is an $[n(k,d,q),k,d]_q$ linear code, then we call $C$ a \textit{length-optimal} linear code. The length-optimal linear code is also distance-optimal, while the converse is not necessarily true.
	If the optimality type of $C$ is determined by the Griesmer bound, we say $C$ is a \textit{Griesmer distance/ length-optimal} linear code.
	When we do not distinguish between optimal types, we call C an \textit{optimal or Griesmer optimal} linear code.

	
	\subsection{Finite geometry}

	Let $\mathcal{P}_k=\{\mathbf{p}_1,\mathbf{p}_2,\ldots,\mathbf{p}_{v_k} \}$ be the set of distinct points in the projective geometry $PG(k-1,q)$, where $v_k=\frac{q^k-1}{q-1}$. 
	A $j$-flat is a projective subspace of dimension $j$ in $PG(k-1,q)$. $0$-flats and $(k-2)$-flats are called points and hyperplanes. 
	\textit{Multi-set} $\aleph$ in $PG(k-1,q)$ 
    is denoted by $\{m_i\cdot\mathbf{p}_i \mid m_i\in  \mathbb{N}_0, \mathbf{p}_i\in \mathcal{P}_k, 1\le i\le v_k\}$. 
     $\aleph$ can also be viewed as a mapping: $\mathcal{P}_k \to  \mathbb{N}_0$. 
     For $\mathbf{p}_i\in \mathcal{P}_k$, $\aleph(\mathbf{p}_i)=m_i$.  
	If $\mathcal{Q}$ is a subset of $\mathcal{P}_k$, then $\aleph(\mathcal{Q})=\sum_{\mathbf{p}\in \mathcal{Q}} \aleph(\mathbf{p})$.
	We define
	\begin{equation}
		[\aleph]=  (\overset{\aleph(\mathbf{p}_1)}{\overbrace{ \mathbf{p}_1,\ldots ,\mathbf{p}_1} },\ldots,
		\overset{\aleph(\mathbf{p}_{v_k})}{\overbrace{ \mathbf{p}_{v_k},\ldots ,\mathbf{p}_{v_k} }}),\; \mathbf{p}_i\in \mathcal{P}_k,
	\end{equation}
	and $\mathrm{Rank}(\aleph)=\mathrm{Rank}([\aleph])$.

	\begin{definition}\cite{landjev2001geometric}
		A multi-set $\aleph$ in $PG(k-1,q)$ is called an $(n,w;k-1,q)$-arc, if
		
		(1)  $\aleph(\mathcal{P}_k)=n$;
		
		(2)  $\aleph(H)\le w$ for any hyperplane $H$;
		
		(3) there exists a hyperplane $H_0$ with $\aleph(H_0)=w$.
	\end{definition}

	The integer $\aleph(\mathbf{p})$ is called the \textit{multiplicity} of point $\mathbf{p}$ and $n$ is called the \textit{cardinality} of multi-set $\aleph$.
	Let $C_{\aleph}$ be the $[n,k,n-w]_q$ linear code generated by $[\aleph]$.
	Then, the linear code $C_{\aleph}$ and the arc $\aleph$ are said to be  \textit{associated}.
	In this paper, the arc $\aleph$ and the associated linear code $C_\aleph$ will be treated equivalently.
	For example, we refer to $C_{\aleph}$ as being composed of points.
	$\mathcal{P}_k$ generates the well-known Simplex code, which has parameters $[v_k,k,q^{k-1}]_q$.
	Let $s\cdot \mathcal{P}_k$ be the multi-set prescribing multiplicity $s$ to every point of $\mathcal{P}_k$.

	\begin{lemma}\cite{macwilliams1977theory}\label{Griesmer_Point_multiplicity}
		If $C$ is an $[n,k,d]_q$ Griesmer code, then any point in $C$ repeats at most $\lceil  \frac{d}{q^{k-1}} \rceil$ times.
	\end{lemma}

	%
		%
	
	\begin{lemma}\cite{ward1999introduction} \label{replicated Simplex code}
		Let $C$ be an $[n,k,d]_q$ linear code. The following statements are equivalent.
		
		(1)  $C$ is a constant weight code.
		
		
		(2)  $C$ is equivalent to replicated Simplex code, possibly with added $0$-coordinates.
		
	\end{lemma}
	
	

	
	\begin{definition}\cite{HAMADA1993229}
		A multi-set $\Im$ in  $PG(k-1, q)$, $k \geq 3 $, is called an  $\{f, m ; k-1, q\}$-minihyper if
		
		(1)  $\Im(\mathcal{P}_k)=f $;
		
		(2)  $\Im(H) \geq m$  for every hyperplane  $H$ ;
		
		(3) there exists a hyperplane  $H_{0}$  with  $\Im\left(H_{0}\right)=m $.
	\end{definition}
	
	Let $\gamma(\Im)= \max\{\Im(\mathbf{p})  \mid  \mathbf{p} \in \mathcal{P}_k\}$ and $s$ be an integer greater or equal to $\gamma(\Im)$. It follows that
	\begin{equation}
		\eth_s (\Im):  \left\{\begin{matrix}
			\mathcal{P}_k \to  \mathbb{N}_0\\
			\mathbf{p} \mapsto  s-\Im(\mathbf{p})
		\end{matrix}\right.
	\end{equation}
	is an $\{s v_{k}-f,s v_{k-1} -m;{k}-1,q\}$-arc.
	If $s=\gamma(\Im)$, then we denote $\eth_s(\Im)$ by $\eth(\Im)$.
	We call the $[sv_{k}-f,k,sq^{k-1}-f+m ]_q$ linear code  $C_{\eth_s (\Im)}$ established by the arc $\eth_s (\Im)$ also \textit{associated} with the minihyper $\Im$.
	If $C_{\eth_s (\Im)}$ is a Griesmer code, then we call minihyper $\Im$ a \textit{Griesmer minihyper}.
	Since the correspondence between minihyper and linear code, we also use $ \mathfrak{e}(\Im)=|\{H\mid H(\Im)=m , H\in PG(k-1,q) \}|$ to denote the error coefficient of linear code $C_{\eth_s (\Im)}$.

        
	
	\subsection{The binary Solomon-Stiffler  and Belov-type minihypers}
	Let $A$ be a non-empty subset of $[k]$. 
 To describe $\mathcal{P}_{|A|}=PG(|A|-1,2)$ more precisely, we define 
	\begin{equation}
		\mathcal{P}_A=\{ (a_1,a_2,\ldots,a_k):a_i \in \mathbb{F}_2, \text{ if }  i \in A \text{, and } a_i=0 \text{ if } i\not\in A\}
	\end{equation}
	as the $|A|$-dimensional subspace associated with $A$ in $PG(k-1,2)$.
	\begin{definition}\cite{solomon1965algebraically}
		Let $A_1$, $A_2$, $\ldots$, $A_{h}$ be $h$ non-empty subsets of $[k]$, with $|A_1|\ge |A_2|\ge \ldots\ge |A_{h}|$.
		Then
		$\Im_{SS}=\bigcup _{i=1}^{h} \mathcal{P}_{A_i}$ is called an \textit{Solomon-Stiffler (SS)-type $\left\{\sum_{i=1}^{h}  v_{|{A_i}|}, \sum_{i=1}^{h} v_{|{A_i}|-1} ; k-1, 2\right\} $-minihyper}.  %
	\end{definition}
	
	\begin{definition}\cite{belov1974construction}\label{Belov_type_minihyper}
		Let $A_1$, $A_2$, $\ldots$, $A_{h}$ be $h$ non-empty subsets of $[k]$, with $|A_1|\ge |A_2|\ge \ldots\ge |A_{h}|= t\ge 4$.
		We define $\mathcal{P}_{T_t}$ to be a set of $t+1$ points in $PG(k-1,2)$ such that $\mathrm{Rank} (\mathcal{P}_{T_t})=t$ and $\sum_{\mathbf{p}\in \mathcal{P}_{T_t}}\mathbf{p}=\mathbf{0}$.
		
		Then
		$\Im_{BV_1}=\bigcup _{i=1}^{h-1} \mathcal{P}_{A_i} \bigcup   (\mathcal{P}_{A_{h}}\setminus \mathcal{P}_{T_t} )$
		and $\Im_{BV_2}=\bigcup _{i=1}^{h-1} \mathcal{P}_{A_i} \bigcup   (\mathcal{P}_{A_{h}}\setminus \mathcal{P}_{T_t} ) \bigcup  \{\mathbf{p} \}$ are called \textit{Belov-type minihypers}
		with parameters $\left\{\sum_{i=1}^{h}  v_{|{A_i}|}-t-1, \sum_{i=1}^{h} v_{|{A_i}|-1} -1; k-1, 2\right\} $ and $\left\{\sum_{i=1}^{h}  v_{|{A_i}|}-t, \sum_{i=1}^{h} v_{|{A_i}|-1} ; k-1, 2\right\} $, where $\mathbf{p}\in PG(k-1,2)$.
	\end{definition}
	
	For an $[n,k,d]_q$ linear code $C$, if $(s-1)2^{k-1}<d \le s2^{k-1}$, then $d$ can be written uniquely as $d=s2^{k-1}-\sum_{i=0}^{k-2}   \lambda_i 2^{i}$.
	We use
	$\bm{\lambda}=(\lambda_0,\ldots,\lambda_{k-2})$ to denote the \textit{$2$-adic anti-expansion vector} of $C$.

	\begin{example}
		Suppose $C$ is a $[11,4,5]_2$ linear code.
		Since $5=2^3-\sum_{i=0}^{1}2^i$, by the above definition, the $2$-adic anti-expansion vector of $C$ is $\bm{\lambda}=(1,1,0)$.
	\end{example}
	
	\begin{lemma} (\cite{helleseth1984further,landjev2007weighted} ) \label{Tho_finite_geometry_results}Let $C$ be an  $[n, k, d]_2$ Griesmer code with $2$-adic anti-expansion vector $\bm{\lambda}$.
		If $wt(\bm{\lambda})= 2$, then the associated minihyper of $C$ is an SS or Belov type minihyper.
	\end{lemma}

	Throughout this paper, since binary AFER-optimal linear codes possess the greatest potential for application, our construction results and examples in this paper are mainly binary.
	Additionally, we agree that $s_1$ and $s_2$ are two integers satisfying $s_1\ge 1$ and $s_2\ge 0$.
	\section{Iterative bounds for error coefficients of Griesmer optimal  linear codes}\label{Sec III}
	
	In this section, we propose five iterative lower bounds on the error coefficients of Griesmer optimal linear codes. We achieve this by analyzing the relationships between linear codes and their residual codes, as well as the structure of related minihypers. Similar to the Griesmer bound, our iterative lower bounds establish a direct connection between high-dimensional and low-dimensional AFER-optimal linear codes. Specifically, our bounds provide lower bounds on the error coefficients of $(k+1)$-dimensional Griesmer optimal linear codes based on the iterative results of AFER-optimal codes with dimensions $\leq k$.

	\subsection{Iterative lower bound I}
	
	In this subsection, we determine the parameters of all two-dimensional AFER-optimal linear codes. Building on this, we present the first iterative lower bound on the error coefficients of linear codes.

	\begin{definition}
		Let $\mathbf{u}=(u_1,u_2,\ldots,u_n)$ be a vector of $\mathbb{F}_q^n$, and let $A  $ be a  subset of $[n]$.
		Then, we define the puncture operation $\Upsilon $ on $\mathbf{u}$ with respect to $A  $ as
		\begin{equation}
			\Upsilon _{A  }(\mathbf{u})=(u_{i_1},u_{i_2},\ldots,u_{i_m}),
		\end{equation}
		where $\{i_1,i_2,\ldots,i_{m}\}=[n]\setminus {A  }$.
		Moreover, for an $[n,k]_q$ linear code $C$, the punctured code with respect to $A  $ is also denoted as
		\begin{equation}
			\Upsilon _{A  }(C)=\{ \Upsilon _{A  }(\mathbf{c}) \mid  \mathbf{c} \in C\}.
		\end{equation}
		
		In particular, when $A  =\mathrm{Supp}(\mathbf{c})$ is the support of a codeword $\mathbf{c}$ of $C$, we also use $	\Upsilon _{\mathbf{c}}(C)$ to denote $	\Upsilon _{A  }(C)$.  $\Upsilon _{\mathbf{c}}(C)$ is also called \textit{residual code} of $C$ under codeword $\mathbf{c}$.
		
		Let ${A  }_C=\mathrm{Supp}(C)$ be the support of $C$. Then, we define
		\begin{equation}
			\Phi  (C)=\{ \Upsilon _{[n]\setminus  {A  }_C}(\mathbf{c}) \mid  \mathbf{c} \in C\}.
		\end{equation}
	\end{definition}

	\begin{theorem}\label{AFER-Two-Dimension}
		Let $s$ and $t$ be two non-negative integers such that $s+t>0$ and $t<q$.
		Then, the $[s(q+1)+t+1,2,sq+t;(q-1)(t+1)]_q$ linear code is AFER-optimal.
	\end{theorem}
	\begin{proof}
		Two-dimensional Simplex code has parameters $[q+1,2,q;q^2-1]_q$.
		The $[s(q+1)+t+1,2,sq+t;(q-1)(t+1)]_q$ Griesmer optimal linear code can be obtained by simple juxtaposition and puncturing of two-dimensional Simplex code.
		Suppose $C$ is an $[s(q+1)+t+1,2,sq+t]_q$ AFER-optimal linear code.
		We prove this theorem in the following cases.
		
		\textbf{Case A:}  If $s=0$ and $0<t<q$, then $C$ is an MDS linear code. According to \cite{macwilliams1977theory}, the $[t+1,2,t]_q$ linear code has unique weight enumerator $W_C(x)=1+(q-1)(t+1)x^t+(q-1)(q-t)x^{t+1}$.
		Therefore, the minimum value of the error coefficient of $C$ is $(q-1)(t+1)$.
		
		\textbf{Case B:} If $s>0$ and $t=0$, then $C$ has at least $q-1$ codewords of minimum weight.
		
		\textbf{Case C:} If $s>0$ and $t>0$, then $C$ is a non-projective linear code.
		Let $\{ \mathbf{p}_1, \mathbf{p}_2, \ldots, \mathbf{p}_{q+1}\} $ be  the $q+1$ points in $PG(1
		,q)$, and let $m_i$ denotes the multiplicity of point $\mathbf{p}_i$ in $C$.
		By Lemma \ref{Griesmer_Point_multiplicity}, we have $m_i\le s+1$.     Since all points in  $PG(1,q)$ form a line and points are also hyperplanes, arbitrary $\delta>1$ different points in $PG(1,q)$ generate an  $[\delta,2,\delta-1]_q$ linear code.
		Therefore, a generator matrix of $C$ can be written as
		
		\begin{equation}
			G=  \left( [(s+1-s^{\prime})\cdot \mathcal{P}_2],G_1,G_{2},\ldots,G_{s^{\prime}} \right),
		\end{equation}
		where $G_i$ consists of different points.
		
		We use $C_i$ and $n_i$ to denote the code generated by $G_i$ and its length, respectively. It follows that $C_i$ is an $[n_i,2,n_i-1]_q$ linear code with weight enumerator $W_{C_i}(x)=1+(q-1)n_ix^{n_i-1}+(q-1)(q-n_i+1)x^{n_i}$.
		For $1\le i<j\le s^{\prime}$, $n_i+n_j> q+1$.
		
		Let $C^{\prime}=C_1\times C_2\times \cdots\times C_{s^{\prime}}$. Then, $C^{\prime}$ is an $[(s^{\prime}-1)(q+1)+t+1,2,(s^{\prime}-1)q+t]_q$ linear code, and has the same error coefficient with $C$.
		In the following, we turn to prove $A_d(C^{\prime})\ge (q-1)(t+1)$ by recursion.
		
		\textbf{Case C.1:}
		If $s^{\prime} =1$, then $C^{\prime}$ is an $[t+1,2,t]_q$ MDS code. By (1), $C^{\prime}$ has unique error coefficient $(q-1)(t+1)$.
		
		\textbf{Case C.2: }
		If $s^{\prime} =2$, then $C^{\prime}$ is a juxtaposition code of $C_1$ and $C_2$, and their weight enumerator are  $W_{C_1}(x)=1+(q-1)n_1x^{n_1-1}+(q-1)(q-n_1+1)x^{n_1}$ and $W_{C_2}(x)=1+(q-1)(n_2)x^{n_2-1}+(q-1)(q-n_2+1)x^{n_2}$, respectively.
		Let $\mathbf{c}=(\mathbf{c}_1,\mathbf{c}_2)$ denote a codeword of $C^{\prime}$, where $\mathbf{c}_i\in C_i$.
		Then $C$ has the minimum error coefficient when all codewords $\mathbf{c}_1$ of weight $n_1$ are juxtaposited with
		codewords $\mathbf{c}_2$ of weight $n_2-1$.
		Therefore, $A_d(C)\ge n_2(q-1)-(q+1-n_1)(q-1)=(n_1+n_2-q-1)(q-1)=(t+1)(q-1)$.
		
		\textbf{Case C.3: }
		If $s^{\prime} \ge 3$, then, by Case C.2, $A_d(C_1\times C_2)\ge (n_1+n_2-q-1)(q-1)$.
		Similarly, we obtain $A_d(C_1\times C_2\times C_3)\ge (n_1+n_2+n_3-2(q-1) )(q-1)$ by using the same method.
		By induction, $A_d(C_1\times C_2\times \cdots \times C_{s^{\prime}})\ge (\sum_{i=1}^{s^{\prime}}n_i-s^{\prime}(q-1) )(q-1)=(t+1)(q-1)$.
		
		In conclusion, we complete the proof.
	\end{proof}

	Since Theorem \ref{AFER-Two-Dimension} determines the parameters of all two-dimensional AFER-optimal linear codes, we focus on linear codes with dimensions $k \geq 3$ in the subsequent discussions.
	
	Building on the results for two-dimensional AFER-optimal linear codes and by analyzing the relationship between linear codes and their residual codes, we present the following theorem.

	\begin{theorem}\label{Enhanced_GriesmerNound_q nmid}
		Let $C$ be an $[n,k,d]_q$ optimal linear code with $k\ge 3$.
		Let $t=q$ if $q\mid d$, and $t=d \mod{q}$ if $q\nmid d$.
		Then, we have
		\begin{equation}\label{Griesmer_type_bound_(2).2}
			A_d(C) \ge L_q^{(1)}(n,k):\triangleq t \left(e_{\left\lceil \frac{d}{q} \right\rceil}(n-d,k-1,q) +\Delta \right) +(q-1)+\Delta^{\prime},
		\end{equation}
		where $\Delta$ is the smallest non-negative integer such that $(q-1) \mid \left( \left\lfloor \frac{ L_q^{(1)}(n,k) -(q-1) }{t} \right\rfloor - \Delta \right)$, and $\Delta^{\prime}$ is the smallest non-negative integer such that $t \mid \left( L_q^{(1)}(n,k)-(q-1) -\Delta^{\prime}\right)$.
	\end{theorem}

	\begin{proof}
		Suppose $C$ is an $[n,k,d]_q$ linear code.
		Let $\mathbf{c}$ be a codeword of $C$ of minimum weight $d$.
		By Theorem 2.7.1 of \cite{huffman2010fundamentals}, we have $\Upsilon _{\mathbf{c}}(C)$ is an $\left[n - d, k - 1, \ge \frac{d}{q}  \right]_q$  linear code.
		Let us prove its error coefficient below.
 
        		We may assume that $C=\left \langle  C^{\prime}, \mathbf{c} \right \rangle $. 
		 According to the previous definition, for any codeword $\mathbf{c}^{\prime}$ of $C^{\prime}$, we have that
		\begin{equation}
			wt(\Upsilon _{\mathbf{c}}(\mathbf{c}^{\prime})) =n(\left \langle \mathbf{c},\mathbf{c}^{\prime}  \right \rangle )-d.
		\end{equation}
		
		Since $\left \langle \mathbf{c},\mathbf{c}^{\prime}  \right \rangle$ is an $[n,2,d]_q$ linear code, $n(\left \langle \mathbf{c},\mathbf{c}^{\prime}  \right \rangle )\ge g_q(2,d)=d+\lceil \frac{d}{q} \rceil$. 
				Let $\mathbf{c}_r=\Upsilon _{\mathbf{c}}(\mathbf{c}^{\prime})$. 
		That is,  $\mathbf{c}_r$ is a codeword of $\Upsilon _{\mathbf{c}}(C)$. 
		Then, it follows that $wt(\mathbf{c}_r)=\lceil \frac{d}{q} \rceil$ if and only if $n(\left \langle \mathbf{c},\mathbf{c}^{\prime}  \right \rangle )=g_q(2,d)$. 
		Thus,  $\Phi  (\left \langle \mathbf{c},\mathbf{c}^{\prime}  \right \rangle)$ is a $[g_q(2,d),2,d]_q$ linear code.
		Therefore, for every codeword $\mathbf{c}_r$ of weight $\lceil \frac{d}{q} \rceil$ in $\Upsilon _{\mathbf{c}}(C)$, we can find a corresponding two-dimensional Griesmer code in $C$, from which we have
		\begin{equation}
			A_d(C)\ge A_{\lceil \frac{d}{q} \rceil}(\Upsilon _{\mathbf{c}}(C))\frac{ e(g_q(2,d),2,q)-(q-1)  }{q-1}  +(q-1).
		\end{equation}
		
		
		By Theorem \ref{AFER-Two-Dimension}, we get $e(g_q(2,d),2,q)=(q-1)(t+1)$.
		Since $e_{\lceil \frac{d}{q} \rceil}(n-d,k-1,q)\le A_{\lceil \frac{d}{q} \rceil}(\Upsilon _{\mathbf{c}}(C))$, we have     \begin{equation}
			e_{\lceil \frac{d}{q} \rceil}(n-d,k-1,q)\le \left\lfloor \frac{A_d(C) -(q-1) }{t} \right\rfloor.
		\end{equation}

		Due to $(q-1)\mid e_{\lceil \frac{d}{q} \rceil}(n-d,k-1,q)$, we can make appropriate transformations to this equation to obtain Equation (\ref{Griesmer_type_bound_(2).2}). This completes the proof of this theorem.
	\end{proof}

\begin{remark}\label{one-step Griesmer bound}
The proof of Theorem \ref{Enhanced_GriesmerNound_q nmid} does not involve Griesmer optimality, so it applies to all optimal linear codes. However, \( e_{\left\lceil \frac{d}{q} \right\rceil}(n-d,k-1,q) \) is greater than zero, and Theorem 4 can perform well, if and only if \( [n-w,k-1,d-\lfloor \frac{w}{q}\rfloor]_q \) is also an optimal code. For example,  \( [9,5,3;4]_2 \) linear code in Table \ref{Five_Binary_AFER} later can be directly determined to be AFER-optimal by Theorem \ref{Enhanced_GriesmerNound_q nmid} and $[6,4,2;3]_2$ AFER-optimal linear code in Table \ref{Four_Binary_AFER}. This method of inferring optimality through residual codes is called the \textit{one-step Griesmer bound} in Grassl's code table \cite{Grassl:codetables}. 
\end{remark}
					
	\begin{example}\label{Ex.first_bound}
		Set $q=2$. 
	 	Let $\mathcal{P}_{I_t}$ be a set of $t$ points in $PG(k-1,2)$ such that $[\mathcal{P}_{I_t}]=I_t$, where $I_t$ is an identity matrix of order $t$. 
		It is easy to verify that $[s_1\cdot \mathcal{P}_{[3]}]$,  $[s_1\cdot \mathcal{P}_{[3]}\bigcup  \mathcal{P}_{I_1}]$, $[s_1\cdot \mathcal{P}_{[3]}\bigcup  \mathcal{P}_{I_2}]$, $[s_2\cdot \mathcal{P}_{[3]}\bigcup  \mathcal{P}_{I_3}]$, $[s_1\cdot \mathcal{P}_{[3]}\setminus  \mathcal{P}_{[2]}]$, $[s_1\cdot \mathcal{P}_{[3]}\setminus  \mathcal{P}_{I_2}]$, $[s_1\cdot \mathcal{P}_{[3]}\setminus  \mathcal{P}_{I_1}]$ generate linear codes with parameters $[7s_1,3,4s_1;7]_2$,  $[7s_1+1,3,4s_1;3]_2$, $[7s_1+2,3,4s_1;1]_2$, $[7s_2+3,3,4s_2+1;3]_2$, $[7s_2+4,3,4s_2+2;6]_2$, $[7s_2+5,3,4s_2+2;2]_2$, $[7s_2+6,3,4s_2+3;4]_2$, respectively.
		According to Theorem \ref{Griesmer_Bound}, those codes are all Griesmer optimal.
		
		Additionally, by Theorems \ref{AFER-Two-Dimension} and \ref{Enhanced_GriesmerNound_q nmid}, we can directly derive that
		\begin{itemize}
			\item $e(7s_1,3,2)\ge 2e_{2s_1}(3s_1,2,2)+1=7$.
			\item $e(7s_1+1,3,2)\ge 2e_{2s_1}(3s_1+1,2,2)+1=3$.
			\item $e(7s_1+2,3,2)\ge 2e_{2s_1}(3s_1+2,2,2)+1=1$.
			\item $e(7s_2+3,3,2)\ge e_{2s_2+1}(3s_2+2,2,2)+1=3$.
			\item $e(7s_2+4,3,2)\ge 2e_{2s_2+1}(3s_2+2,2,2)+1=5$.
			\item $e(7s_2+5,3,2)\ge e_{2s_2+1}(3s_2+3,2,2)+1=1$.
			\item $e(7s_2+6,3,2)\ge e_{2s_2+2}(3s_2+3,2,2)+1=4$.
		\end{itemize}

		Therefore, $[7s_1,3,4s_1;7]_2$,  $[7s_1+1,3,4s_1;3]_2$, $[7s_1+2,3,4s_1;1]_2$, $[7s_2+3,3,4s_2+1;3]_2$, $[7s_2+6,3,4s_2+3;4]_2$ linear codes are AFER-optimal, which also demonstrates that Theorem \ref{Enhanced_GriesmerNound_q nmid} is tight in these cases.
		For $[7s_2+4,3,4s_2+2;6]_2$ and $[7s_2+5,3,4s_2+2;2]_2$ linear codes, the lower bound provided by Theorem 4 differs from their error coefficients by only $1$.
	\end{example}
	
	In the next subsections, we will give another lower bounds, which can determine that $[7s_2+4,3,4s_2+2;6]_2$ and $[7s_2+5,3,4s_2+2;2]_2$ linear codes are also AFER-optimal.
	
	\subsection{Iterative lower bound II}
				When restricting the Griesmer optimal types of linear codes, we can get another lower bound on the error coefficients of Griesmer optimal linear codes.

	Firstly, to precisely describe the Griesmer optimality of linear codes, we provide the following definition.
	
	\begin{definition}\label{def_Griesmer optimal}
		For an $[n,k,d]_q$ Griesmer optimal linear code, we define
		$\Gamma _q(n,k,d)= k$ for $ n = g_q(k, d) $. Otherwise, let
		$ \Gamma _q(n,k,d) $ be the smallest integer $k_1$ such that $ n - g_q(k_1, d) \geq g_q\left(k - k_1, \left\lceil \frac{d}{q^{k_1}} \right\rceil + 1 \right)$. 
	\end{definition}
	
	
	The following lemma characterizes the properties of Griesmer optimal linear codes.
	
	\begin{lemma}\label{Lem_critical length-optimal}
		Assume $C$ is an $[n,k,d]_q$ Griesmer optimal linear code with $k\ge 3$.
		Then, the following statements hold.
		
		(1)  If $\Gamma _q(n,k,d)< k$, then $q\mid d$.
		
		(2)  If $q\nmid d$, then $C$ is a Griesmer code.
		
	\end{lemma}
	\begin{proof}
		As $C$ is a Griesmer optimal linear code, we have $n\ge g_q(k,d)$ and $n<g_q(k,d+1)$.
		Since         \begin{equation}
			g_q(k,d+1)-g_q(k,d)= \sum\limits_{i=0}^{k-1}\left\lceil  \frac{d+1}{q^{i}}\right\rceil-\sum\limits_{i=0}^{k-1}\left\lceil  \frac{d}{q^{i}}\right\rceil =\left\{\begin{matrix}
				1,& q\nmid d, \\
				2,& q\mid d,
			\end{matrix}\right.
		\end{equation}
		it is easy to see (1) and (2) hold.
	\end{proof}

	\begin{definition}
		Let $C$ be an $[n,k,d]_q$ linear code and $\mathbf{v}$ be a vector of $\mathbb{F}_q^n$. Define
		\begin{equation}
			C+\mathbf{v}=\{ \mathbf{c}+  \mathbf{v}\mid  \mathbf{c}\in C\}.
		\end{equation}
	\end{definition}
	
	Based on the constraints of minimum-weight codewords, the following proposition is proposed.
	
	\begin{proposition}\label{Pro_SubCode_Residual}
		Suppose $C$ is an $[n,k,d]_q$ Griesmer optimal linear code.
		If $q\mid d$ and $\Gamma _q(n,k,d)\ge 2$, then for any codeword $\mathbf{c}$ of weight $d$ of $C$, we have
		\begin{equation}\label{Eq_Eq_minimum_linear_nonlinear_subcode}
			A_\frac{d}{q}(\Upsilon _{\mathbf{c}}(C))\le \min\{A_d(C^{\prime}),\min\{ A_d(C^{\prime}+ \alpha \mathbf{c})\mid \alpha\in \mathbb{F}_q^*\}\},
		\end{equation}
		where $C^{\prime}$ is a supplementary subcode of $\mathbf{c}$ such that $C=\left \langle \mathbf{c},C^{\prime}  \right \rangle $.
	\end{proposition}
	\begin{proof}
		Since $q\mid d$ and $\Gamma _q(n,k,d)\ge 2$, $\Upsilon _{\mathbf{c}}(C)$ is an $[n-d,k-1,\frac{d}{q}]_q$ Griesmer-optimal linear code.
		Let $\mathbf{c}^{\prime}$ be a codeword of $C^{\prime}$, and let $m=\max\{wt(\mathbf{c}^{\prime}), wt(\mathbf{c}^{\prime}+\alpha \mathbf{c})\mid \alpha\in \mathbb{F}_q^*\}$.
		According to Lemma 	\ref{Modifications_Griesmer_Bound}, we have \begin{equation}
			n(\Upsilon _{\mathbf{c}}(\left \langle \mathbf{c},\mathbf{c}^{\prime} \right \rangle ))-d\ge \mathsf{g}_q(2, d,m)-d=\left\lceil \frac{m}{q}\right\rceil.
		\end{equation}
		This formula takes the equality if and only if $m=d$.
		By Lemma \ref{replicated Simplex code}, at this time, $\left \langle \mathbf{c},\mathbf{c}^{\prime}  \right \rangle$ is a two-dimensional replicated Simplex code with possible $0$-coordinates.
		That is, $wt(\mathbf{c}^{\prime})=wt(\mathbf{c}^{\prime}+ \alpha \mathbf{c})=d$ holds for all $\alpha$ in $\mathbb{F}_q^*$. It is natural to conclude that Equation (\ref{Eq_Eq_minimum_linear_nonlinear_subcode}) is true.
	\end{proof}

	By characterizing the relationship between Griesmer distance-optimal linear code and its puncture and residual codes, we deduce the following theorem.
	
	\begin{theorem}\label{Theo_Extend_Code}
		If $C$ is an $[n,k,d]_q$ Griesmer optimal linear code with  $\Gamma _q(n,k,d)< k$.
		then
		\begin{equation}
			A_d(C)\ge L_q^{(2)}(n,k):\triangleq e(n-d-1,k-1,q).
		\end{equation}
	\end{theorem}
	\begin{proof}
		By Lemma \ref{Lem_critical length-optimal}, we have that $q\mid d$ and $n> g_q(k,d)$. 
		It follows that $d(n-1,k,q)=d$.
		Let us prove the minimum error coefficient of $C$ in two different scenarios.

		(1) If puncture some position of $C$ yields  an $[n-1,k,d]_q$ linear code $C^{\prime}$, then we have
		\begin{equation}
			A_d(C)\ge  \min\{A_d(\tilde{C} ^{\prime})\mid  \tilde{C} ^{\prime} \text{ is a $(k-1)$-dimensional  linear subcode of } C^{\prime} \}.
		\end{equation}

		Let $\mathbf{c}^{\prime}$ be a codeword of weight $d$ of $C^{\prime}$. As $\Gamma _q(n-1,k,d)\ge 2$, we have  $\Upsilon _{\mathbf{c}^{\prime}}(C^{\prime})$ is an $[n-d-1,k-1,\frac{d}{q}]_q$ Griesmer optimal linear code.
		By Proposition \ref{Pro_SubCode_Residual}, we get
		\begin{equation}
			A_{\frac{d}{q}}\left( \Upsilon _{\mathbf{c}^{\prime}}(C^{\prime}) \right)\le \min\{A_d(\tilde{C} ^{\prime})\mid  \tilde{C} ^{\prime} \text{ is a $(k-1)$-dimensional  linear subcode of } C^{\prime} \}.
		\end{equation}
		
		As $e(n-d-1,k-1,q)\le e\left( \Upsilon _{\mathbf{c}^{\prime}}(C^{\prime}) \right)$, we can directly deduce that $A_d(C)\ge e(n-d-1,k-1,q)$.
		
		(2) If $C$ is not extended from an $[n,k,d]_q$ linear code, then puncture any position on $C$ deduce an $[n-1,k,d-1]_q$ linear code $C^{\prime\prime}$.
		Let $\mathbf{c}^{\prime\prime}$ be a codeword of weight $d-1$ of $C^{\prime \prime}$.
		As $q\nmid (d-1)$, $\Upsilon _{\mathbf{c}^{\prime\prime}}(C^{\prime\prime})$ is also an $[n-d-1,k-1,\frac{d}{q}]_2$ Griesmer optimal linear code.
		According to the proof of Theorem \ref{Enhanced_GriesmerNound_q nmid}, we have
		\begin{equation}
			A_{d-1}(C^{\prime\prime}) \ge (q-1)A_{\frac{d}{q}}(\Upsilon _{\mathbf{c}^{\prime\prime}}(C^{\prime\prime}))+(q-1)\ge (q-1)e(n-d-1,k-1,q)+(q-1).
		\end{equation}
		
		Since $C^{\prime\prime}$ is obtained by puncturing $C$, it follows that $A_d(C)\ge A_{d-1}(C^{\prime\prime})$.
		Thus, in this case, we have  $A_d(C)\ge (q-1)e(n-d-1,k-1,q)+(q-1)$.
	\end{proof}
	
	In Example \ref{Ex.first_bound}, the AFER-optimality of $[7s_2+5,3,4s_2+2;2]_2$ linear code remains undetermined. As the following example shows, its AFER-optimality can be solved by Theorem \ref{Lem_critical length-optimal}.
	\begin{example}
		Set $q=2$.
		Since $[7s_2+5,3,4s_2+2]_2$ linear code is Griesmer optimal and $\Gamma _2(7s_2+5,3,4s_2+2)=1< 3$, by Theorem \ref{Lem_critical length-optimal}, we have
		$e(7s_2+5,3,2)\ge e(3s_2+2,2,2)=2$.
		Therefore, $[7s_2+5,3,4s_2+2;2]_2$ linear code is also AFER-optimal.
	\end{example}
	
	\subsection{Iterative lower bound III}
	
	In this subsection, we present a lower bound on the error coefficients of linear codes by analyzing the relationships between subcode structures of linear codes.

		%
		%

	According to Definition \ref{def_Griesmer optimal}, we can deduce that Griesmer optimal linear codes have a specific subcode structure, as shown in the following proposition.

	\begin{proposition}\label{Generalized_Griesmer_code_code_chain}
		Let $C$ be an $[n,k,d]_q$ Griesmer optimal linear code, and let $k_1=\min\{\Gamma _q(n,k,d),k-1\}$.
		Then, $C$ has a linear subcode chain $C_{1}    \subset C_{2}\subset \cdots \subset C_{{k_1}}\subset C$ satisfies
		\begin{equation}\label{E_Code_Chain}
			\begin{array}{cc}
				n(C_{1} )=g_q(1,d), \\
				n(C_{2})=g_q(2,d),\\
				\vdots\\
				n(C_{{k_1}})=g_q(k_1,d),
			\end{array} 
		\end{equation}
		where $C_{i}$ is an $i$-dimensional linear subcode of $C$.
	\end{proposition}
	\begin{proof}
Before we formally begin the proof, we give a description of the punctured code of $C$ according to multiple codewords. 
    Let $k^{\prime}\le k_1$ be a positive integer and 
		let $\mathbf{c}_1$, $\mathbf{c}_2$, $\ldots$, $\mathbf{c}_{k^{\prime}}$ be $k^{\prime}$ linearly independent codewords of $C$.
       Suppose $A  =\mathrm{Supp}(\left \langle  \mathbf{c}_1,\mathbf{c}_2,\ldots,\mathbf{c}_{k^{\prime}} \right \rangle ) $.
		Then, it follows that, for all $\mathbf{c}_{r}\in \Upsilon _{A}(C)$, we have
		\begin{equation}
			wt(\mathbf{c}_{r})\ge d_r=\min \left\{
			n(\left \langle  \mathbf{c}_1,\mathbf{c}_2,\ldots,\mathbf{c}_{k^{\prime}},\mathbf{c} \right \rangle ) \mid  \mathbf{c}\in C\setminus \left \langle  \mathbf{c}_1,\mathbf{c}_2,\ldots,\mathbf{c}_{k^{\prime}} \right \rangle \right\}
			-
			n(\left \langle  \mathbf{c}_1,\mathbf{c}_2,\ldots,\mathbf{c}_{k^{\prime}} \right \rangle ).
		\end{equation}
		This is, $\Upsilon _{A}(C)$ is an $[n-|A  |,k-k^{\prime},d_r]_q$ linear code. 
		In the following, we prove that the proposition holds in two cases.
		
		\textbf{Case A}: When $ n = g_q(k, d) $.  
		Let $\mathbf{c}_{1}$ be a codeword of weight $d$ of $C$ and $A  _{1}=\mathrm{Supp}(\mathbf{c}_{1})$.
		Then,  $\Upsilon _{A  _1}(C)$ is an $[n-d,k-1$ linear code with minimum distance
		$\min\{n(\left \langle \mathbf{c}_{1},\mathbf{c}_{2}  \right \rangle )-d\mid \mathbf{c}_{2}\in C \setminus \left \langle  \mathbf{c}_1 \right \rangle   \}.$
		
		According to Theorem \ref{Griesmer_Bound}, for all codeword $\mathbf{c}_{2}$ of $C\setminus \left \langle  \mathbf{c}_1 \right \rangle $, $n(\left \langle \mathbf{c}_{1},\mathbf{c}_{2}  \right \rangle ) \ge g_q(2,d)$.
		Since $n-d=g_q\left(k-1,\left\lceil \frac{d}{q} \right\rceil\right)$, the $\left[n-d,k-1, \left\lceil \frac{d}{q} \right\rceil + 1\right]_q$ linear code violates the Griesmer bound. Thus, there must exist some codeword $\mathbf{c}_{2}$ of $C\setminus \left \langle  \mathbf{c}_1 \right \rangle $ such that $n(\left \langle \mathbf{c}_{1},\mathbf{c}_{2}  \right \rangle ) = g_q(2,d)$.
		By induction, we can find another linearly independent codewords $\mathbf{c}_{3}$, $\mathbf{c}_{4}$, $\ldots$, $\mathbf{c}_{{k-1}}$ of $C$ such that \begin{equation}
			n(\left \langle  \mathbf{c}_1,\mathbf{c}_2,\ldots,\mathbf{c}_{k^{\prime}} \right \rangle )=g_{q}(k^{\prime},d)
		\end{equation}
		holds for all $2\le k^{\prime}\le k-1=k_1$.
		Denote the linear codes $\left \langle \mathbf{c}_{1} \right \rangle $, $\left \langle \mathbf{c}_{1},\mathbf{c}_{2} \right \rangle $, $\ldots$,  $\left \langle \mathbf{c}_{1},\mathbf{c}_{2},\ldots,\mathbf{c}_{{k-1}} \right \rangle $ by $C_{1}$, $C_{2}$, $\ldots$, $C_{{k-1}}$.
		Then, Equation (\ref{E_Code_Chain}) can be derived immediately.
		
		\textbf{Case B}: When $ n > g_q(k, d) $.  We have $ k_1 $ is the smallest integer such that $ n - g_q(k_1, d) \geq g_q(k - k_1, \lceil \frac{d}{q^{k_1}} \rceil + 1) $ and $k_1\ge1$.
		Therefore, the $\left[n-g_{q}(k^{\prime},d),k-k^{\prime}, \left\lceil \frac{d}{q^{k^{\prime}}} \right\rceil  +1\right]_q$ linear code violates the Griesmer bound for $1\le k^{\prime}\le k_1-1$.
		Similar with Case A, we can find $k_1$ linearly independent codewords of $\mathbf{c}_{1}$, $\mathbf{c}_{2}$, $\ldots$, $\mathbf{c}_{{k_1}}$ of $C$ such that
		\begin{equation}
			n(\left \langle  \mathbf{c}_1,\mathbf{c}_2,\ldots,\mathbf{c}_{k^{\prime}} \right \rangle )=g_{q}(k^{\prime},d),
		\end{equation}
		holds for all $2\le k^{\prime}\le k_1$.
		This also implies Equation (\ref{E_Code_Chain}) holds.
	\end{proof}
	
	Relying on Propositions \ref{Pro_SubCode_Residual} and \ref{Generalized_Griesmer_code_code_chain}, the following theorem can be established.
	
	\begin{theorem}\label{Th. Enhanced_GriesmerNound_q mid d_based}
		Let $C$ be an $[n,k,d]_q$ Griesmer optimal linear code with $q\mid d$ and $k_1=\min\{\Gamma _q(n,k,d),k-1\}\ge 2$.
		Let $s=\left \lceil \frac{d}{q^{k-1}} \right\rceil $, and let $\mu_q(n,k) =\max\{ e(g_q(k_1,d),k_1,q),e_d(n-s,k-1,q)\}$.
		Then, we have
		\begin{equation}\label{Griesmer_type_bound_(2).2_based}
			e(n,k,q)\ge L_q^{(3)}(n,k):\triangleq (q-1) (e(n-d,k-1,q)+\Delta) +\mu_q(n,k)+(q-1).
		\end{equation}
		where $\Delta$ is the smallest non-negative integer such that $(q-1) \mid \left( \left\lfloor\frac{L_q^{(3)}(n,k)-\mu_q(n,k)-(q-1)}{q-1}\right\rfloor-\Delta \right)$.
	\end{theorem}
	\begin{proof} 
                		Since $k_1=\min\{\Gamma _q(n,k,d),k-1\}\ge 2$, the residual code of $C$ under codeword of weight $d$ has parameters $[n-d,k-1,\frac{d}{q}]_q$. 	
                With Proposition \ref{Generalized_Griesmer_code_code_chain}, there must exist some codeword $\mathbf{c}$ of weight $d$ in $C\setminus  C_{k-1}$ such that $C=\left \langle \mathbf{c},  C_{k-1} \right \rangle $. 
		It follows that \begin{equation}
			A_d(C)=A_d(C_{k-1})+\sum\limits_{i=0}^{q-2} A_d(C_{k-1}+\alpha^{i}\mathbf{c})+q-1.
		\end{equation} 
        It is easy to conclude that there must be at least one subcode in $ C_{k-1}+\mathbf{c}$, $\ldots$, and $C_{k-1}+\alpha^{q-2}\mathbf{c}$ with error coefficient less than  $\left\lfloor\frac{A_d(C)-A_d(C_{k-1})-(q-1)}{q-1}\right\rfloor$. 
        In the following, we complete the proof of this theorem by proving that the lower bound on the maximum error coefficient of $(k-1)$-dimensional linear subcode of $C$ is $\mu_q(n,k)$. 
    
		According to Proposition \ref{Generalized_Griesmer_code_code_chain}, $C$ has a $k_1$-dimensional linear subcode $C_{{k_1}}$ satisfies $n(C_{{k_1}})=g_q(k_1,d)$.
		Without loss of generality, assume that $C_{k-1}$ is a $(k-1)$-dimensional linear subcode of $C$ such that $C_{{k_1}}\subset C_{k-1}$.
	This implies $A_d(C_{k-1})\ge e(g_q(k_1,d),k_1,q)$.

        In addition, as there exists point $\mathbf{p}$ in $C$ repeats at least $s$ times, $C$ has a $(k-1)$-dimensional linear subcode $C_{k-1}^{\prime}$ satisfies $n(C_{k-1}^{\prime})\le n-s$. 
		Moreover, for the reason that $g_q(k-1,d+1)=g_q(k,d+1)-\left \lceil \frac{d+1}{q^{k-1}} \right\rceil\ge g_q(k,d+1)-s\ge n-s$, $\Phi (C_{k-1}^{\prime})$ is also an $[\le n-s,k-1,d]_q$ Griesmer optimal linear code. 
        Therefore, $C$ has a $(k-1)$-dimensional subcode has error coefficients at least $\mu_q(n,k)$.

		By Proposition \ref{Pro_SubCode_Residual}, we have
		
		\begin{equation}
			A_\frac{d}{q}(\Upsilon _{\mathbf{c}}(C))\le \left\lfloor\frac{A_d(C)-\max\{A_d(C_{k-1}),A_d(C_{k-1}^{\prime})\}-(q-1)}{q-1}\right\rfloor=\left\lfloor\frac{A_d(C)-\mu_q(n,k)-(q-1)}{q-1}\right\rfloor.
		\end{equation}

		Since  $(q-1)\mid A_\frac{d}{q}(\Upsilon _{\mathbf{c}}(C))$, we can make appropriate transformations to this equation to obtain Equation (\ref{Griesmer_type_bound_(2).2_based}).
		Hence, this theorem is true.
	\end{proof}
	
	Under Theorem \ref{Th. Enhanced_GriesmerNound_q mid d_based}, we can prove that
	$[7s_2+4,3,4s_2+2;6]_2$ linear code is AFER-optimal, as shown in the following example.
	\begin{example}
		Let the notation be the same with Example \ref{Ex.first_bound}.
		Since $[7s_2+4,3,4s_2+2;6]_2$ linear code $C_1$ is a Griesmer code, $\Gamma _2(7s_2+4,3,4s_2+2)=3\ge 2$.
		It follows that $C$ has a $2$-dimensional subcode with effective length $ g_2(2,4s_2+2)$.
		According to Theorem \ref{Th. Enhanced_GriesmerNound_q mid d_based}, we have $e(7s_2+4,3,2)\ge e(6s_2+3,2,2)+e_{2s_2+1}(3s_2+2,2,2)+1=6$.
		Thus, $[7s_2+4,3,4s_2+2;6]_2$ linear code is AFER-optimal.
	\end{example}

	\subsection{Iterative lower bound IV} 	
	
	Using Proposition \ref{Generalized_Griesmer_code_code_chain}, we derive another lower bound on the error coefficients of Griesmer optimal linear codes, as presented in the following theorem.
	
	
	\begin{theorem}\label{k_Residue_Grismer_Bound}
		Let $C$ be an $[n,k,d]_q$ Griesmer optimal linear code with $k_1=\min\{\Gamma _q(n,k,d),k-1\}\ge 2$.
		Let $k_2<k_1$ be a positive integer.
		Define $\varsigma_q(k_2,d) =e(g_q(k_2+1,d),k_2+1,q) - e(g_q(k_2,d),k_2,q) $.
		Then, we have
		\begin{equation}\label{Eq_Improved_Griesmer_type_bound_k_Residue}
			A_d(C)\ge L_q^{(4)}(n,k):\triangleq
			\max\left\{
			\frac{\left( e(n-g_q(k_2,d),k-k_2,q)+ \Delta\right)\varsigma_q(k_2,d)}{q-1}
			+ e(g_q(k_2,d),k_2,q)+\Delta^{\prime}
			\mid 1\le k_2< k_1\right\},
		\end{equation}
		where $\Delta$ is the smallest non-negative integer such that $(q-1) \mid  \left( \left\lfloor \frac{ L_q^{(4)}(n,k) - e(g_q(k_2,d),k_2,q) }{\varsigma_q(k_2,d)} \right\rfloor-\Delta \right)$ and
		$\Delta^{\prime}$ is the smallest non-negative integer such that $\varsigma_q(k_2,d) \mid \left( L_q^{(4)}(n,k) - e(g_q(k_2,d),k_2,q) -\Delta^{\prime}\right)$.
	\end{theorem}
	\begin{proof}
		According to Proposition \ref{Generalized_Griesmer_code_code_chain}, we can find $k_1$ bases vectors  $\mathbf{c}_{1}$, $\mathbf{c}_{2}$, $\ldots$, $\mathbf{c}_{{k_1}}$ of $C$, such that \begin{equation}
			n(\left \langle  \mathbf{c}_1,\mathbf{c}_2,\ldots,\mathbf{c}_{k_2} \right \rangle )=g_{q}(k_2,d)
		\end{equation}
		holds for all $2\le k_2< k_1$.
		This is, $\left \langle  \mathbf{c}_1,\mathbf{c}_2,\ldots,\mathbf{c}_{k_2} \right \rangle$ is a $[g_{q}(k_2,d),k_2,d]_q$ Griesmer code.
		Let $C_{{k_2}}=\left \langle  \mathbf{c}_1,\mathbf{c}_2,\ldots,\mathbf{c}_{k_2} \right \rangle$ and $A  =\mathrm{Supp}(C_{{k_2}})$.
		Since $k_2<k_1$, the $\left[n-g_{q}(k_2,d),k-k_2,\left\lceil \frac{d}{q} \right\rceil+1\right]_q$ linear code violates the Griesmer bound, and $\Upsilon _{A  }(C) $ is an $\left[n-g_{q}(k_2,d),k-k_2,\left\lceil \frac{d}{q} \right\rceil\right]_q$ Griesmer code.
		It follows that, for any codeword $\mathbf{c}$ of $C\setminus C_{{k_2}}$, $\Upsilon _{A  }(\mathbf{c})$ has weight $\left\lceil \frac{d}{q^{k_2}} \right\rceil$ if and only if $n(\left \langle  \mathbf{c}_1,\mathbf{c}_2,\ldots,\mathbf{c}_{k_2},\mathbf{c} \right \rangle)=g_q(k_2+1,d)$.
		At this time, $\left \langle  \mathbf{c}_1,\mathbf{c}_2,\ldots,\mathbf{c}_{k_2},\mathbf{c} \right \rangle$ is an $[g_{q}(k_2+1,d),k_2+1,d]_q$ Griesmer code with error coefficient at least $e(g_q(k_2+1,d),k_2+1,q)$.
		Therefore, we have the following constrains relationship
		\begin{equation}
			\frac{e(n-g_q(k_2,d),k-k_2,q)}{A_d(C) - e(g_q(k_2,d),k_2,q)}	\le
			\frac{q-1}{e(g_q(k_2+1,d),k_2+1,q) - e(g_q(k_2,d),k_2,q)}=\frac{q-1}{\varsigma_q(k_2,d)}.
		\end{equation}
		
		Since $(q-1)\mid e\left(\Upsilon _{A  }(C) \right)$, we can directly derive that,
		\begin{equation}
			A_d(C)\ge	\frac{\left( e(n-g_q(k_2,d),k-k_2,q)+ \Delta\right)\varsigma_q(k_2,d)}{q-1}
			+ e(g_q(k_2,d),k_2,q)+\Delta^{\prime}.
		\end{equation}
		As this formula holds for $1\le k_2\le k_1$,  we get Equation (\ref{Eq_Improved_Griesmer_type_bound_k_Residue}).
		Hence, we complete the proof.
	\end{proof}

	The implementation of Theorem \ref{k_Residue_Grismer_Bound} depends on the special subcode structure of Griesmer optimal linear codes.
	If $k_2=1$, then Theorems \ref{Enhanced_GriesmerNound_q nmid} and \ref{k_Residue_Grismer_Bound} are the same. 
	When we have more information about AFER-optimal linear codes of low dimensions, Theorem \ref{k_Residue_Grismer_Bound} may perform better than Theorems \ref{Enhanced_GriesmerNound_q nmid} and \ref{Th. Enhanced_GriesmerNound_q mid d_based}, as shown in the following example.
	
	\begin{example}
		Set $q=2$.
		According to Theorem \ref{Griesmer_Bound}, $[15s_2+12,4]_2$ Griesmer optimal linear code $C$ has minimum distance $8s_2+6$.
		Since $\Gamma _2(15s_2+12,4,8s_2+6)=4\ge 2$, by Theorems \ref{Enhanced_GriesmerNound_q nmid} and \ref{Th. Enhanced_GriesmerNound_q mid d_based}, we have $$e(15s_2+12,4,2)\ge \max\{2e_{4s_2+3}(7s_2+6,3,2)+1,e_{4s_2+3}(7s_2+6,3,2)+e(14s_2+11,3,2)+1 \}=\max\{9,11\}=11.$$
		
		By Theorem \ref{k_Residue_Grismer_Bound}, we have
		$$e(15s_2+12,4,2)\ge e(3s_2+3,2,2)\varsigma_2(2,8s_2+6)+e(12s_2+8,3,2)=12 $$ for $k_2=2$.
		Therefore, at this time, Theorem \ref{k_Residue_Grismer_Bound} has better performance.
	\end{example}


	\subsection{Iterative lower bound V}
	Since Griesmer codes and Griesmer minihypers have a one-to-one correspondence, we provide a lower bound for the error coefficients of Griesmer codes by characterizing Griesmer minihypers.
	
	The following lemma reveals how the minihyper's rank acts on the error coefficient.
	
	\begin{lemma}\label{minihyper_kissing_number}
		Suppose $\Im$  is an $\{f, m ; k-1, q\}$-minihyper with rank $k^{\prime}<k$.
		Let $\Im^{\prime}$ be an $\{f, m ; k^{\prime}-1, q\}$-minihyper such that $ \Im^{\prime}= \Im \bigcap PG(k^{\prime}-1,q)$ and $\mathrm{Rank}(\Im^{\prime})=k^{\prime}$.
		Then, we have
		\begin{equation}
			\mathfrak{e}(\Im)=q^{k-k^{\prime}}\mathfrak{e}(\Im^{\prime}).
		\end{equation}
	\end{lemma}
	\begin{proof}
		Since $\Im$ is an $\{f, m ; k-1, q\}$-minihyper with rank $k^{\prime}<k$, after the appropriate row transformation, we have
		\begin{equation}
			[\Im]=\left(  \begin{array}{cc}
				[\Im^{\prime}]   \\
				\mathbf{0}_{(k-k^{\prime})\times f}
			\end{array} \right).
		\end{equation}
		It follows that $\mathfrak{e}(\Im)=q^{k-k^{\prime}}\mathfrak{e}(\Im^{\prime})$.
	\end{proof}
	
	Given minihyper with special parameters, we establish an upper bound on its rank, as shown in the following lemma.
	\begin{lemma}\label{Lem_special_minihyper}
		The maximum rank of $\{n,1;k-1,q\}$-minihyper $\Im$ is $n-1$.
	\end{lemma}
	\begin{proof}
		If $\Im$ has rank $n$, then $\Im$ will be an $\{n,0;k-1,q\}$-minihyper. This leads to a contradiction.
	\end{proof}
	
	In \cite{helleseth1984further,hamada1985characterization,HAMADA1993229,landjev2007weighted}, the authors characterized properties of special Griesmer minihypers, which will also help us to determine the maximum rank of associated minihyper.
	If the rank of Griesmer minihyper is known, then we have the following theorem.
	
	\begin{theorem}\label{Theo_minihyper_rank_bound}
		Let $C$ be an $[n,k,d]_q$ Griesmer code, and let $s=\left\lceil  \frac{d}{q^{k-1}} \right\rceil$.
		If the $\{sv_k-n, sv_{k-1}-n+d; k-1, q\}$-minihyper $\Im$ has maximum rank $k^{\prime}<k$, then we have
		\begin{equation}
			A_d(C)\ge L_q^{(5)}(n,k):\triangleq q^{k-k^{\prime}}e(n+sv_{k^{\prime}}-sv_k ,k^{\prime},q).
		\end{equation}
	\end{theorem}
	\begin{proof}
		Since $C$ is a Griesmer code, by Lemma \ref{Griesmer_Point_multiplicity}, all points in $C$ repeat at most $s$ times. Thus, $\eth(C)$ is an $\{sv_k-n, sv_{k-1}-n+d; k-1, q\}$-minihyper $\Im$.
		By Lemma \ref{minihyper_kissing_number}, if the $\{sv_k-n, sv_{k-1}-n+d; k-1, q\}$-minihyper $\Im$ has maximum rank $k^{\prime}<k$, then  we have $\mathfrak{e}(\Im)=q^{k-k^{\prime}}\mathfrak{e}(\Im^{\prime})$, where $ \Im^{\prime}= \Im \bigcap PG(k^{\prime}-1,q)$ and $\mathrm{Rank}(\Im^{\prime})=k^{\prime}$.
		It follows that \begin{equation}
			A_d(C)=\mathfrak{e}(\eth(C))=q^{k-k^{\prime}}\mathfrak{e}(\Im^{\prime})\ge q^{k-k^{\prime}}e(n+sv_{k^{\prime}}-sv_k ,k^{\prime},q).
		\end{equation}
	\end{proof}

	\begin{example}
		Set $q=2$. 
		Suppose $C$ is a $[15s_2+11,4,8s_2+5]_2$ linear code.
		Since $C$ is a Griesmer code, the associated minihyper $\Im$ of $C$ has parameters $\{ 4,1;3,2\}$.
		By Lemma \ref{Lem_special_minihyper}, $\mathrm{Rank}(\Im)\le 3$.
		According to Theorem \ref{Theo_minihyper_rank_bound}, we get $e(15s_2+11,4,2)\ge 2 e(7s_2+3,3,2)=6$.
	\end{example}
	
	\begin{remark}
		Our proposed bounds are iterative, progressing from low dimensions to high dimensions.
		\begin{itemize}
			\item General Iterative Bound: Theorem \ref{Enhanced_GriesmerNound_q nmid} provides a general iterative bound without any restrictions.
			\item Constrained Iterative Bounds: Theorems \ref{Theo_Extend_Code}, \ref{Th. Enhanced_GriesmerNound_q mid d_based}, \ref{k_Residue_Grismer_Bound}, and \ref{Theo_minihyper_rank_bound} require constraints on specific optimal types.
		\end{itemize}
		During the iteration process, information from low-dimensional AFER-optimal codes or associated Griesmer minihypers is also required. 
		The more information we obtain, the tighter the lower bounds on the error coefficients of Griesmer optimal linear codes provided by Theorems \ref{Enhanced_GriesmerNound_q nmid}-\ref{Theo_minihyper_rank_bound}.
	\end{remark}

	To help readers in utilizing the five bounds more conveniently, we present the following theorem. This theorem is directly derived by summarizing Theorems \ref{Enhanced_GriesmerNound_q nmid} through \ref{Theo_minihyper_rank_bound}. 
		Notably, the theorems present a progressive structure, differing primarily in their optimality and assumptions about known low-dimensional AFER-optimal linear codes. There is scope overlap in their applicability, for example when $k_2=1$, Theorems \ref{Enhanced_GriesmerNound_q nmid} and \ref{k_Residue_Grismer_Bound} agree, but the latter can derive better lower bounds on the error coefficient for other values of $k_2$. 
	Five theorems are structured in Theorem \ref{Theo_Mix_lower_bound} for utility: Theorems \ref{Enhanced_GriesmerNound_q nmid}-\ref{Theo_Extend_Code} have minimal assumptions on known  AFER-optimal linear codes, providing lower bounds on the error coefficients for all (Griesmer) optimal linear codes (though possibly loose). With more known information about AFER-optimal linear codes, Theorems \ref{Th. Enhanced_GriesmerNound_q mid d_based}–\ref{Theo_minihyper_rank_bound} yield tighter bounds.

	

	\begin{theorem}\label{Theo_Mix_lower_bound}
		With notation as above, let \( C \) be an \( [n,k,d]_q \) optimal linear code for $k\ge 3$. Then
		\begin{equation}\label{Eq_five_lower_bounds}
			e(n,k,q) \ge \max \left\{
			\begin{array}{cl}
				L_q^{(1)}(n,k), & \text{unconditional (Theorem \ref{Enhanced_GriesmerNound_q nmid})}, \\
				L_q^{(2)}(n,k), & \Gamma_q(n,k,d) < k \text{ (Theorem \ref{Theo_Extend_Code})}, \\
				L_q^{(3)}(n,k), & q \mid d \text{ and } \Gamma_q(n,k,d) \ge 2 \text{ (Theorem \ref{Th. Enhanced_GriesmerNound_q mid d_based})}, \\
				L_q^{(4)}(n,k), & \Gamma_q(n,k,d) \ge 2 \text{ (Theorem \ref{k_Residue_Grismer_Bound})}, \\
				L_q^{(5)}(n,k), & n = g_q(k,d) \text{ and the maximum rank of the associated Griesmer minihyper is known (Theorem \ref{Theo_minihyper_rank_bound})},
			\end{array}
			\right.
		\end{equation}
	where the first bound applies to any optimal linear code but performs better for those that are optimal under the one-step Griesmer bound (Remark \ref{one-step Griesmer bound}); the remaining four bounds apply only to Griesmer optimal codes, whose specific optimality is characterized by \( \Gamma_q(n,k,d) \).
	\end{theorem}


	Numerical results indicate that our bounds are tight for binary linear codes, as shown in Tables \ref{Three_Binary_AFER}, \ref{Four_Binary_AFER}, and \ref{Five_Binary_AFER}. These tables include generator matrices or multisets of related AFER-optimal linear codes. Additional codes are listed in Appendix. Specifically, our bounds are tight for dimensions $k < 5$. When $k = 5$, our bounds are not tight in only a few cases, and these codes are not included in Table \ref{Five_Binary_AFER}. In the next section, we will determine the parameters of these unlisted AFER-optimal linear codes and evaluate the performance of our bounds in these cases.




	\begin{remark}
		It should be noted that the bounds presented in this section can directly determine that Griesmer optimal linear codes in \cite{abdullah2023some,abdullah2023new} are AFER-optimal.
		According to the calculation results of linear programming bounds listed in Table I of \cite{abdullah2023some} and Table II of \cite{abdullah2023new}, our bounds are tighter than those of linear programming bound \cite{sole2021linear} in those cases.
	\end{remark}

	\begin{table*}[ht]
		\caption{Three-dimensional AFER-optimal binary linear codes with $s_1\ge 1$ and $s_2\ge 0$}\label{Three_Binary_AFER}
		\centering
				\begin{tabular}{cccc}
					\toprule
					Lengths  &       Parameters        &                    Constructions                     & AFER-Optimality  \\ \midrule
					$7s_1$  &   $[7s_1,3,4s_1;7]_2$   &               $s_1\cdot \mathcal{P}_{[3]}$               & Theorem \ref{Theo_Mix_lower_bound},  Case 1          \\
					$7s_1+1$ &  $[7s_1+1,3,4s_1;3]_2$  &         $s_1\cdot \mathcal{P}_{[3]}\bigcup  \mathcal{P}_{I_1}$         &  Theorem \ref{Theo_Mix_lower_bound},  Case 1      \\
					$7s_1+2$ &  $[7s_1+2,3,4s_1;1]_2$  &         $s_1\cdot \mathcal{P}_{[3]}\bigcup  \mathcal{P}_{I_2}$         &  Theorem \ref{Theo_Mix_lower_bound},  Case 1      \\
					$7s_2+3$ & $[7s_2+3,3,4s_2+1;3]_2$ &         $s_2\cdot \mathcal{P}_{[3]}\bigcup  \mathcal{P}_{I_3}$         & Theorem \ref{Theo_Mix_lower_bound},  Case 1                      \\
					$7s_2+4$ & $[7s_2+4,3,4s_2+2;6]_2$ & $s_1\cdot \mathcal{P}_{[3]}\setminus  \mathcal{P}_{[2]}$ & Theorem \ref{Theo_Mix_lower_bound},  Case 3                               \\
					$7s_2+5$ & $[7s_2+5,3,4s_2+2;2]_2$ &    $s_1\cdot \mathcal{P}_{[3]}\setminus  \mathcal{P}_{I_2}$    &               Theorem \ref{Theo_Mix_lower_bound},  Case 2                 \\
					$7s_2+6$ & $[7s_2+6,3,4s_2+3;4]_2$ &    $s_1\cdot \mathcal{P}_{[3]}\setminus  \mathcal{P}_{I_1}$    &         Theorem \ref{Theo_Mix_lower_bound},  Case 1                       \\ \bottomrule
				\end{tabular}
		\end{table*}

		\begin{table*}[ht]
			\caption{Four-dimensional AFER-optimal binary linear codes with $s_1\ge 1$ and $s_2\ge 0$}\label{Four_Binary_AFER}
			\centering
			\begin{tabular}{cccc}
				\toprule
				Lengths   &         Parameters         &                                    Constructions                                    &             AFER-Optimality              \\ \midrule
				$15s_1$   &   $[15s_1,4,8s_1;15]_2$    &                              $s_1\cdot \mathcal{P}_{[4]}$                               &          Theorem \ref{Theo_Mix_lower_bound},  Case 1          \\
				$15s_1+1$  &   $[15s_1+1,4,8s_1;7]_2$   &                        $s_1\cdot \mathcal{P}_{[4]}\bigcup  \mathcal{P}_{I_1}$                         &          Theorem \ref{Theo_Mix_lower_bound},  Case 1          \\
				$15s_1+2$  &   $[15s_1+2,4,8s_1;3]_2$   &                        $s_1\cdot \mathcal{P}_{[4]}\bigcup  \mathcal{P}_{I_2}$                         &          Theorem \ref{Theo_Mix_lower_bound},  Case 1          \\
				$15s_1+3$  &   $[15s_1+3,4,8s_1;1]_2$   &                        $s_1\cdot \mathcal{P}_{[4]}\bigcup  \mathcal{P}_{I_3}$                         &          Theorem \ref{Theo_Mix_lower_bound},  Case 1          \\
				$15s_2+4$  &  $[15s_2+4,4,8s_2+1;4]_2$  &                        $s_2\cdot \mathcal{P}_{[4]}\bigcup  \mathcal{P}_{I_4}$                         &          Theorem \ref{Theo_Mix_lower_bound},  Case 1          \\
				$15s_2+5$  & $[15s_2+5,4,8s_2+2;10]_2$  &                        $s_2\cdot \mathcal{P}_{[4]}\bigcup  \mathcal{P}_{T_4}$                         & Theorem \ref{Theo_Mix_lower_bound},  Case 3 \\
				$15s_2+6$  &  $[15s_2+6,4,8s_2+2;3]_2$  &        $s_2\cdot \mathcal{P}_{[4]}\bigcup  \mathcal{P}_{T_4}^{\prime}$       &           Theorem \ref{Theo_Mix_lower_bound},  Case 2          \\
				$15s_2+7$  &  $[15s_2+7,4,8s_2+3;7]_2$  & $s_1\cdot \mathcal{P}_{[4]}\setminus (\mathcal{P}_{[3]} \bigcup  \mathcal{P}_{\{4\}}) $ &          Theorem \ref{Theo_Mix_lower_bound},  Case 1          \\
				$15s_2+8$  & $[15s_2+8,4,8s_2+4;14]_2$  &                $s_1\cdot \mathcal{P}_{[4]}\setminus  \mathcal{P}_{[3]}$                 & Theorem \ref{Theo_Mix_lower_bound},  Case 3 \\
				$15s_2+9$  &  $[15s_2+9,4,8s_2+4;6]_2$  &      $s_1\cdot \mathcal{P}_{[4]}\setminus  (\mathcal{P}_{[3]} \setminus \mathcal{P}_{I_1}) $       &             Theorem \ref{Theo_Mix_lower_bound},  Case 2            \\
				$15s_2+10$ & $[15s_2+10,4,8s_2+4;2]_2$  &      $s_1\cdot \mathcal{P}_{[4]}\setminus  (\mathcal{P}_{[3]} \setminus \mathcal{P}_{I_2}) $       &           Theorem \ref{Theo_Mix_lower_bound},  Case 2           \\
				$15s_2+11$ & $[15s_2+11,4,8s_2+5;6]_2$  &      $s_1\cdot \mathcal{P}_{[4]}\setminus  (\mathcal{P}_{[3]} \setminus \mathcal{P}_{I_3}) $       &         Theorem \ref{Theo_Mix_lower_bound},  Case 5          \\
				$15s_2+12$ & $[15s_2+12,4,8s_2+6;12]_2$ &                $s_1\cdot \mathcal{P}_{[4]}\setminus  \mathcal{P}_{[2]} $                 &             Theorem \ref{Theo_Mix_lower_bound},  Case 4            \\
				$15s_2+13$ & $[15s_2+13,4,8s_2+6;4]_2$  &              $s_1\cdot \mathcal{P}_{[4]}\setminus  \mathcal{P}_{I_2} $               &           Theorem \ref{Theo_Mix_lower_bound},  Case 2            \\
				$15s_2+14$ & $[15s_2+14,4,8s_2+7;8]_2$  &              $s_1\cdot \mathcal{P}_{[4]}\setminus  \mathcal{P}_{I_1} $               &           Theorem \ref{Theo_Mix_lower_bound},  Case 1           \\ \bottomrule
			\end{tabular}
			\begin{tablenotes}    
				\footnotesize               
				\item[1]	To simplify, $[\mathcal{P}_{I_t}]=I_t$, $\mathcal{P}_{T_t}=\mathcal{P}_{I_t}\bigcup  \{ \mathbf{1}_t^T\}$, and
				$\mathcal{P}_{T_4}^{\prime}=\mathcal{P}_{I_4}\bigcup  \{ (1,0,1,1)^T,(0,1,1,1)^T\}$.
			\end{tablenotes}            
		\end{table*}
		

		\begin{table*}[ht]
			\caption{Five-dimensional AFER-optimal binary linear codes with $s_1\ge 1$ and $s_2\ge 0$}\label{Five_Binary_AFER}
			\centering	
			\begin{tabular}{cccc}
				\toprule
				Length   &         Parameters          &                                  Constructions                                  &               AFER-Optimality               \\\midrule
				8      &        $[8,5,2;1]_2$         &                          $G_{[8,5,2;1]_2}$                           &                                                Theorem \ref{Theo_Mix_lower_bound},  Case 1      \\
				12     &        $[12,5,4;1]_2$        &                         $	G_{[12,5,4;1]_2}$                     &            Theorem \ref{Theo_Mix_lower_bound},  Case 1            \\
				9      &        $[9,5,3;4]_2$        &                                $G_{[9,5,3;4]_2}$                                &                    Theorem \ref{Theo_Mix_lower_bound},  Case 1                    \\
				$31s_1$     &   $[31s_1,5,16s_1;31]_2$    &                            $s_1\cdot  \mathcal{P}_{[5]}$                             &           Theorem \ref{Theo_Mix_lower_bound},  Case 1            \\
				$31s_1+1$    &  $[31s_1+1,5,16s_1;15]_2$   &                      $s_1\cdot  \mathcal{P}_{[5]}\bigcup  \mathcal{P}_{I_1}$                       &           Theorem \ref{Theo_Mix_lower_bound},  Case 1            \\
				$31s_1+2$    &   $[31s_1+2,5,16s_1;7]_2$   &                      $s_1\cdot  \mathcal{P}_{[5]}\bigcup  \mathcal{P}_{I_2}$                       &           Theorem \ref{Theo_Mix_lower_bound},  Case 1            \\
				$31s_1+3$    &   $[31s_1+3,5,16s_1;3]_2$   &                      $s_1\cdot  \mathcal{P}_{[5]}\bigcup  \mathcal{P}_{I_3}$                       &           Theorem \ref{Theo_Mix_lower_bound},  Case 1            \\
				$31s_1+4$    &   $[31s_1+4,5,16s_1;1]_2$   &                      $s_1\cdot  \mathcal{P}_{[5]}\bigcup  \mathcal{P}_{I_4}$                       &           Theorem \ref{Theo_Mix_lower_bound},  Case 1            \\
				$31s_2+5$    &  $[31s_2+5,5,16s_2+1;5]_2$  &                      $s_2\cdot  \mathcal{P}_{[5]}\bigcup  \mathcal{P}_{I_5}$                       &           Theorem \ref{Theo_Mix_lower_bound},  Case 1            \\
				$31s_2+6$    &  $[31s_2+6,5,16s+2;15]_2$   &                      $s_2\cdot  \mathcal{P}_{[5]}\bigcup  \mathcal{P}_{T_5}$                       &  Theorem \ref{Theo_Mix_lower_bound},  Case 3    \\
				$31s_2+10$   & $[31s_2+10,5,16s_2+4;10]_2$ &                $([s_2\cdot \mathcal{P}_{[5]}], G_{[10,5,4;10]_2})$                 &             Theorem \ref{Theo_Mix_lower_bound},  Case 2                \\
				$31s_1+11$   & $[31s_1+11,5,16s_1+4;3]_2$ &                 $([s_2\cdot \mathcal{P}_{[5]}],G_{[42,5,20;3]_2})$                 &           Theorem \ref{Theo_Mix_lower_bound},  Case 2                \\
				$31s_2+14$ &  $[31s_2+14,5,16s_2+6;7]_2$  &         $([s_2\cdot \mathcal{P}_{[5]}], G_{[14,5,6;7]_2})$         &          Theorem \ref{Theo_Mix_lower_bound},  Case 2 $^\ast $      \\
				$31s_2+15$ &   $[31s_2+15,5,16s_2+7;15]_2$   & $s_1\cdot \mathcal{P}_{[5]}\setminus (\mathcal{P}_{[4]}\bigcup  \mathcal{P}_{\{ 5\}})$ &          Theorem \ref{Theo_Mix_lower_bound},  Case 1            \\
				$31s_2+16$ &   $[31s_2+16,5,16s_2+8;30]_2$   &                $s_1\cdot \mathcal{P}_{[5]}\setminus \mathcal{P}_{[4]}$                 &  Theorem \ref{Theo_Mix_lower_bound},  Case 3    \\
				$31s_2+17$ &   $[31s_2+17,5,16s_2+8;14]_2$   &         $s_1\cdot \mathcal{P}_{[5]}\setminus (\mathcal{P}_{[4]} \setminus \mathcal{P}_{I_1}) $         &              Theorem \ref{Theo_Mix_lower_bound},  Case 2              \\
				$31s_2+18$ &   $[31s_2+18,5,16s_2+8;6]_2$    &         $s_1\cdot \mathcal{P}_{[5]}\setminus (\mathcal{P}_{[4]} \setminus \mathcal{P}_{I_2}) $         &              Theorem \ref{Theo_Mix_lower_bound},  Case 2                \\
				$31s_2+19$ &   $[31s_2+19,5,16s_2+8;2]_2$    &         $s_1\cdot \mathcal{P}_{[5]}\setminus (\mathcal{P}_{[4]} \setminus \mathcal{P}_{I_3}) $         &            Theorem \ref{Theo_Mix_lower_bound},  Case 2              \\
				$31s_2+22$   &   $[31s_2+22,5,16s_2+10;6]_2$   &   $s_1\cdot \mathcal{P}_{[5]} \setminus (\mathcal{P}_{[4]} \setminus \mathcal{P}_{T_k}^{\prime})  $   &              Theorem \ref{Theo_Mix_lower_bound},  Case 2                \\
				$31s_2+24$   &  $[31s_2+24,5,16s_2+12;28]_2$   &                $s_1\cdot \mathcal{P}_{[5]}\setminus \mathcal{P}_{[3]}$                 &             Theorem \ref{Theo_Mix_lower_bound},  Case 4               \\
				$31s_2+25$   &  $[31s_2+25,5,16s_2+12;12]_2$   &       $s_1\cdot  \mathcal{P}_{[5]}\setminus (\mathcal{P}_{[3]}\setminus \mathcal{P}_{I_1})$       &              Theorem \ref{Theo_Mix_lower_bound},  Case 2                \\
				$31s_2+26$   &   $[31s_2+26,5,16s_2+12;4]_2$   &       $s_1\cdot  \mathcal{P}_{[5]}\setminus (\mathcal{P}_{[3]}\setminus \mathcal{P}_{I_2})$       &              Theorem \ref{Theo_Mix_lower_bound},  Case 2                \\
				$31s_2+27$ & $[31s_2+27,5,16s_2+13;12]_2$ &    $s_1\cdot \mathcal{P}_{[5]}\setminus (\mathcal{P}_{[2]}\bigcup  \mathcal{P}_{\{ 3\}} )$     &         Theorem \ref{Theo_Mix_lower_bound},  Case 4          \\
				$31s_2+28$   &  $[31s_2+28,5,16s_2+14;24]_2$   &                $s_1\cdot  \mathcal{P}_{[5]}\setminus \mathcal{P}_{[2]}$                 &              Theorem \ref{Theo_Mix_lower_bound},  Case 4               \\
				$31s_2+29$   &   $[31s_2+29,5,16s_2+14;8]_2$   &              $s_1\cdot  \mathcal{P}_{[5]}\setminus \mathcal{P}_{I_2} $               &             Theorem \ref{Theo_Mix_lower_bound},  Case 2             \\
				$31s_2+30$ &  $[31s_2+30,5,16s_2+15;16]_2$   &              $s_1\cdot  \mathcal{P}_{[5]}\setminus  \mathcal{P}_{I_1} $               &             Theorem \ref{Theo_Mix_lower_bound},  Case 1             \\ \bottomrule
			\end{tabular}
			\begin{tablenotes}    
				\footnotesize               
				\item[1]	$^\ast $	By Case 2 of Theorem \ref{Theo_Mix_lower_bound}, we have $e(13(s_1+1)+14,5,2)\ge e(15(s_1+1)+7,4,2)\ge 7$. It follows $e(14,5,2)\ge 7$.
			\end{tablenotes}
		\end{table*}

		\section{Binary AFER-optimal linear codes with special parameters}\label{Sec V}		
		In this section, we utilize the properties of binary linear codes and the classification results of special Griesmer codes to determine the parameters of 5-dimensional AFER-optimal linear codes not listed in Table \ref{Five_Binary_AFER}. Additionally, we compare how our bounds differ from the actual values in these cases. Throughout this section, we assume that $q=2$ and $v_k=2^{k}-1$.

		\subsection{Lower bounds on non-extendable binary linear codes}

		This subsection has two main objectives. First, we establish the necessary and sufficient conditions for binary linear codes with even minimum distance to be non-extendable. Second, we derive lower bounds on the error coefficients of binary non-extendable linear codes with even minimum distance.

		\begin{lemma}\label{Lem_non-extendable_even_codes}
			An $[n,k,d]_2$ linear code with $2\mid d$ is non-extendable if and only if $\Upsilon _{\mathbf{c}}(C)$ has parameters $[n-d,k-1,\frac{d}{2}]_2$ for all codewords $\mathbf{c}$ of weight $d$ of $C$.
		\end{lemma}
		\begin{proof}
			Let $C^{\prime}$ be the complement of $\mathbf{c}$ such that $C=\left \langle C^{\prime}, \mathbf{c} \right \rangle $.
			Let $\mathbf{c}^{\prime}$ be a codeword of $C^{\prime}$.
			Then $C$ is  non-extendable if and only if for all codewords $\mathbf{c}$ of weight $d$ of $C$, there must exist some codeword $\mathbf{c}^{\prime}$ of weight $d$ in $C^{\prime}$ such that $wt(\mathbf{c}+\mathbf{c}^{\prime})=d$. That is to say, no matter which linear extension operation is applied to \( C \), there will always be a codeword of minimum weight that cannot be extended.

			If $wt(\mathbf{c}^{\prime})>d$ or $wt(\mathbf{c}^{\prime}+\mathbf{c})>d$, then $n(\left \langle \mathbf{c},\mathbf{c}^{\prime} \right \rangle )\ge \mathsf{g}_2\left(2,d,\max\{ wt(\mathbf{c}^{\prime}),wt(\mathbf{c}^{\prime}+\mathbf{c})\}\right) >g_2(2,d)$.
			It follows that $\Upsilon _{\mathbf{c}}(\mathbf{c}^{\prime})> n(\left \langle \mathbf{c},\mathbf{c}^{\prime} \right \rangle )-d>\frac{d}{2}$.
			Similarly, $\Upsilon _{\mathbf{c}}(\mathbf{c}^{\prime})=\frac{d}{2}$ if and only if $wt(\mathbf{c}^{\prime})=wt(\mathbf{c}^{\prime}+\mathbf{c})=d$, i.e., $\Phi (\left \langle \mathbf{c},\mathbf{c}^{\prime} \right \rangle)$ is a two-dimensional replicated Simplex code.
			Therefore, if $\Upsilon _{\mathbf{c}}(C)$ has parameters $[n-d,k-1,\frac{d}{2}]_2$ for all codewords $\mathbf{c}$ of weight $d$ of $C$, then there always exist some codewords $\mathbf{c}^{\prime}$ in $C$ such that $wt(\mathbf{c}+\mathbf{c}^{\prime})=d$.

			Thus, we finish the proof.
		\end{proof}
		
		\begin{remark}
			In \cite{maruta2001extendability}, Maruta  provided a sufficient condition for binary linear codes with even-distances to be non-extendable.
			Therefore, Lemma \ref{Lem_non-extendable_even_codes} can be seen as an improvement on Maruta's results.
		\end{remark}
		
		Next, let us determine the lower bound on the error coefficients for binary even-distance non-extendable linear codes.

		\begin{definition}
			Suppose $n,k$, and $d$ are three positive integers such that $d\le d(n,k,2)$. We define
			\begin{equation}
				\underline{e}_d(n,k,2):=\min \{A_d(C)\mid \text{ $C$ is an $[n,k,d]_2$ linear code } \}.
			\end{equation}
		\end{definition}

		\begin{theorem}\label{L_extendability_Linear_codes}
			Let $C$ be an $[n,k,d]_2$ linear code with $2\mid d$, and let $s=\left\lceil \frac{n}{2^{k}-1} \right\rceil$.
			If $C$ is non-extendable, then we have $A_d(C)\ge  \max \{ e_d(n-s,k-1,2)+2, e_d(n-s,k-1,2)+\underline{e}_{\frac{d}{2}}(n-d,k-1,2)+1\}$.
		\end{theorem}
		\begin{proof}
			Since $C$ is non-extendable, by Lemma \ref{Lem_non-extendable_even_codes}, for all codewords $\mathbf{c}$ of weight $d$ of $C$, $\Upsilon_{\mathbf{c}}(C)$ has parameters $[n-d,k-1,\frac{d}{2}]_2$.

			As $|PG(k-1,2)|=2^{k-1}$, there must exist some point has multiply $s$ in $C$.
			Without loss of generality, we may assume that a generator matrix of $C$ is
			\begin{equation}
				G=
				\left(
				\begin{array}{cc}
					\mathbf{1}_{s}&\mathbf{x}\\
					\mathbf{0}&G_1
				\end{array}
				\right),
			\end{equation}
			where $wt((\mathbf{1}_{s},\mathbf{x}))=d$.
			It follows that $G_1$ generates an $[n-s,k-1,\ge d]_2$ linear code $C_1$.
			According to Proposition \ref{Pro_SubCode_Residual}, $A_{\frac{d}{2}}(\Upsilon_{\mathbf{c}}(C))\le \min\{A_d(C^{\prime}), A_d(C^{\prime}+ \mathbf{c})\}$ and
			$A_{\frac{d}{2}}(\Upsilon_{\mathbf{c}}(C))\ge \underline{e}_{\frac{d}{2}}(n-d,k-1,2)\ge 1$.
			Thus, we have $A_d(C)\ge A_d(C_1)+A_{\frac{d}{2}}(\Upsilon_{\mathbf{c}}(C))+1\ge  \max \{ e_d(n-s,k-1,2)+2, e_d(n-s,k-1,2)+\underline{e}_{\frac{d}{2}}(n-d,k-1,2)+1\}$.
			Hence, this theorem is true.
		\end{proof}
		
		With the help of Theorem  \ref{L_extendability_Linear_codes}, we prove that $[31s_2 +7,5,16s_2 +2;5]_2$ linear code is AFER-optimal, as shown in the following theorem.
		
		\begin{theorem}\label{Theo_[7,5,2]}
			The	$[31s_2 +7,5,16s_2 +2;5]_2$ linear code is AFER-optimal.
		\end{theorem}
		\begin{proof}
			Let $C$ be a $[31s_2 +7,5,16s_2 +2;e]_2$ AFER-optimal linear code.
			By case 4 of Theorem \ref{Theo_Mix_lower_bound}, we have $e\ge 4$.
			It is easy to check that $([s_2 \cdot \mathcal{P}_{[5]}],G_{[7,5,2;5]_2})$ generates a $[31s_2 +7,5,16s_2 +2;5]_2$ linear code, where $G_{[7,5,2;5]_2}$ is listed in the appendix. Thus, we have $4\le e \le 5$.
			
			According to Lemma \ref{Griesmer_Point_multiplicity}, any point in $C$ repeats at most $s_2 +2$ times.
			If $C$ has some point $\mathbf{p}$ repeats $s+2$ times, then shorten related coordinates on $C$ will deduce a $[30s_2 +5,4,16s_2 +2]_2$ linear code.
			By case 1 of Theorem \ref{Theo_Mix_lower_bound}, we have $e(30s_2 +5,4,2)\ge 10>5$.
			Therefore, any point in $C$ has multiply less than $s_2 +2$.
			We will prove that the $[31s_2 +7,5,16s_2 +2;5]_2$ linear code is AFER-optimal in two cases.
			
			\textbf{Case A:}
			If $C$ is non-extendable, then by Theorem \ref{L_extendability_Linear_codes}, we have $e\ge e_{16s_2 +2}(30s_2 +6,4,2)+2\ge 5$.
			By \cite{Grassl:codetables}, there is no $[8,5,3]_2$ linear code.
			Therefore, $C$ is non-extendable for $s_2=0$ and  $e\ge 5$.
			
			\textbf{Case B:}
			If C is extendable, then $s_2\ge 1$.
			Because the extended code of $C$ has an odd distance, we can obtain an $[31s_2+9,5,16s_2+4]_2$ linear code after two extensions.
			According to Lemmas \ref{Tho_finite_geometry_results} and  \ref{Lem_Belov-type}, there only exist two SS-type $[31s_2+9,5,16s_2+4]_2$ linear codes have different weight enumerators, with associated minihypers $\Im_1=\mathcal{P}_{[4]}\bigcup  \mathcal{P}_{[3,5]}$ and $\Im_2=\mathcal{P}_{[4]}\bigcup  \mathcal{P}_{[2,4]}$.
			Since any points in $C$ have multiply not exceeding $s+1$, we get $C$ by puncturing the two $[31s_2+9,5,16s_2+4]_2$ linear codes.
			After traversal puncturing two $[40,5,20]_2$ linear codes associated with $\Im_1$ and $\Im_2$, we obtain only $[38,5,18;5]_2$ linear code. Thus, $e\ge 5$ at this point.
		\end{proof}

		\subsection{Binary AFER-optimal linear codes form Solomon-Stiffler and Belov Codes}
		
		In this subsection, we mainly determine AFER-optimal Griesmer codes that are SS or Belov Codes with the help of classification results in \cite{helleseth1984further}.
		We assume that $C$ is an  $[n, k, d]_2$ Griesmer code and with $2$-adic anti-expansion vector $\bm{\lambda}=(\lambda_0,\ldots,\lambda_{k-2})$.
		
		The following two lemmas give the restrictive conditions that Griesmer codes are SS or Belov-type linear codes.
		
		\begin{lemma}\cite{belov1974construction}\label{Lemma_not_SS_type}
			If $\sum_{i=0}^{k-2} \lambda_i (i+1)>  k\left\lceil \frac{d}{2^{k-1}} \right\rceil $, then $C$ is not an SS-type linear code.
		\end{lemma}

		\begin{lemma}\label{Lem_Belov-type}
			If $C$ is a Belov-type linear code, then there exists a positive integer $k_{bv}\ge 4$ such that $\lambda_{k_{bv}-1}=0$ and $\lambda_i=1$ for $1\le i\le k_{bv}-2$.
		\end{lemma}
		\begin{proof}
			Let $\Im$ be the associated minihyper of $C$.
			According to Definition \ref{Belov_type_minihyper}, $\Im$ has a sub-multi-set $\Im^{\prime}$ contains all $i$-flats in $\Im$ that are not $\{ v_{i+1}, v_i;i,q\}$-arcs.
			Since $\eth_{k_{bv}}(\Im^{\prime})$ generates a $[k_{bv}+1,k_{bv},2]_2$ or $[k_{bv},k_{bv},1]_2$ linear code, $d\pmod{2^{k_{bv}}}=2^{k_{bv}-1}-2$ or $2^{k_{bv}-1}-1$. It follows that $\lambda_{k_{bv}-1}=0$ and $\lambda_i=1$ for $1\le i\le k_{bv}-2$.
		\end{proof}

		\begin{theorem}\label{Theo_SS_h=2}
			Let $k\ge 4$, $k_1$, and $k_2$ be three positive integers such that $k_2<k_1<k$. Then, the following statements hold.
			
			(1)
			If $k_1\ge k_2+2$ or $k_2\ge 3$,  then, $s_1\cdot \mathcal{P}_{[k]}\setminus (\mathcal{P}_{[k_1]}\bigcup  \mathcal{P}_{[k-k_2+1,k]})$ generates an $[s_1v_k-v_{k_1}-v_{k_2},k,s_12^{k-1}-2^{k_1-1}-2^{k_2-1}]_2$ AFER-optimal linear code with error coefficient
			\begin{equation}
				\left\{
				\begin{array}{cc}
					2^{k}+2^{k-|A_1|-| A_2|}- 2^{k-|A_1|} - 2^{k-|A_2|},& \text{ if } |A_1|+|A_2|\le k. \\
					2^{k}-2^{k-|A_1|} -2^{k-| A_2|}+1 ,& \text{ if } |A_1|+|A_2|> k.
				\end{array}
				\right.
			\end{equation}
			
			(2) For $k\ge 4$, $s_1\cdot \mathcal{P}_{[k]}\setminus (\mathcal{P}_{[4]}\setminus \mathcal{P}_{T_4})$ generates an $[s_1v_k-10,k,s_12^{k-1}-6;5\cdot 2^{k-3}]_2$ AFER-optimal linear code.
		\end{theorem}
		\begin{proof}
			(1)
			According to Lemma \ref{Lem_Belov-type}, when $k_1\ge k_2+2$ or $k_2\ge 3$, the $[s_1v_k-v_{k_1}-v_{k_2},k,s_12^{k-1}-2^{k_1-1}-2^{k_2-1}]_2$ is a SS-type Griesmer code.
			For $A\in [k]$ and $\mathbf{x}\in \mathbb{F}_2^k$, let $\mathbf{x}_{A}$ be a vector obtained from $\mathbf{x}$ by puncturing at the coordinates in $[k]\setminus  A$.
			We put $\mathbf{x}_{A}=\mathbf{0}$ whenever $\mathbf{x}=\mathbf{0}$ or $A=\emptyset$.
			For a multi-set of points $\aleph$ in $PG(k-1,2)$, we define $\sigma_{\aleph}(\mathbf{x})=\mathbf{x} [\aleph]$.
			
			Since $\Im=\mathcal{P}_{A_1}\bigcup  \mathcal{P}_{A_2}$, we have
			\begin{equation}
				wt(\sigma_{\Im}(\mathbf{x}))= \left\{
				\begin{array}{ll}
					0,                          & \text{ if } \mathbf{x}_{A_1}=\mathbf{0},\mathbf{x}_{A_2}=\mathbf{0};                                     \\
					2^{|A_1|-1} ,               & \text{ if } \mathbf{x}_{A_1}\ne \mathbf{0},\mathbf{x}_{A_2}=\mathbf{0};                                  \\
					2^{|A_2|-1} ,               & \text{ if } \mathbf{x}_{A_1}=\mathbf{0},\mathbf{x}_{A_2}\ne\mathbf{0};                                   \\
					2^{|A_1|-1} +2^{|A_2|-1} ,  & \text{ if } \mathbf{x}_{A_1}\ne\mathbf{0},\mathbf{x}_{A_2}\ne\mathbf{0}, \mathbf{x}_{A_1\bigcap A_2}=\mathbf{0};    \\
					2^{|A_1|-1} +2^{|A_2|-1}  , & \text{ if } \mathbf{x}_{A_1}\ne\mathbf{0},\mathbf{x}_{A_2}\ne\mathbf{0}, \mathbf{x}_{A_1\bigcap A_2}\ne \mathbf{0}.
				\end{array}\right.
			\end{equation}
			Consequently, we have
			\begin{equation}
				\begin{array}{rl}
					e(\Im)=& 2^{k-|A_1|-|A_2|+|A_1\bigcap A_2|}\left(
					2^{|A_1| +| A_2|-2|A_1\bigcap A_2|} \left(2^{|A_1\bigcap A_2|}-1 \right)
					+\left(2^{A_1-|A_1\bigcap A_2|}-1 \right) \left(2^{A_2-|A_1\bigcap A_2|}-1 \right)
					\right)\\
					=& 2^{k-|A_1|-| A_2|}  \left(  2^{|A_1|+|A_2|}+2^{|A_1 \bigcap A_2|}-2^{|A_1|}-2^{|A_2|}\right). \\
				\end{array}
			\end{equation}
			
			It follows that $\Im$ has the minimum error coefficient if and only if $|A_1 \bigcap A_2|$ takes the smallest value.
			As $|A_1|+|A_2|-k \le  |A_1 \bigcap A_2| \le |A_2|$ and $|A_1 \bigcap A_2|\ge 0$, the smallest error coefficient of $\Im$ is $2^{k}+2^{k-|A_1|-| A_2|}- 2^{k-|A_1|} - 2^{k-|A_2|}$ for $|A_1|+|A_2|\le k $, and $2^{k}-2^{k-|A_1|} -2^{k-| A_2|}+1$ for $|A_1|+|A_2|> k $.
			
			(2)
			For $k=4$, the Belov-type linear code with minihyper $\mathcal{P}_{[4]}\setminus \mathcal{P}_{T_4}$ has parameters $[15s_2+5,4,8s_2+2;10]_2$.
			As $C$ is a Griesmer code, by Case 2 of Theorem \ref{Theo_Mix_lower_bound}, $e(15s_2+5,4,2)\ge e_{4s_2+1}(7s_2+3,3,2)+e(14s_2+4,3,2)+1\ge 10$.
			Thus, $C$ is AFER-optimal.

			For $k\ge 5$, since $[s_1v_k-10,k,s_12^{k-1}-6]_2$ linear code $C$ satisfies both Lemmas \ref{Lemma_not_SS_type} and \ref{Lem_Belov-type},  $C$ is an SS or Belov-type code.
			If $C$ is an SS-type linear code, then by the proof of (1), we have $e(C)=21\cdot 2^{k-5}$.
			If $C$ is a Belov-type linear code, the associated minihyper of $C$ is $\mathcal{P}_{[4]}\setminus \mathcal{P}_{T_4}$.
			By Lemma \ref{minihyper_kissing_number}, it follows that $e(C)=5\cdot 2^{k-3}$.
			Since $5\cdot 2^{k-3}<21\cdot 2^{k-5}$, we have (2) also holds.
		\end{proof}

		\begin{theorem}\label{Theorem_[21,5,10]}
			The $[31s_1+8,5,16s_1+3;13]_2$,  $[31s_1+12,5,16s_1+5;11]_2$, and $[31s_2+20,5,16s_2+9;8]_2$ linear codes are AFER-optimal.
		\end{theorem}
		\begin{proof}
			According to the Griesmer bound, the three codes are all Griesmer codes.
			Firstly, it is easy to use Magma \cite{bosma1997magma} to check that  $([s_2\cdot \mathcal{P}_{[5]}],G_{[43,5,21;11]_2})$,  $([s_2\cdot \mathcal{P}_{[5]}],G_{[39,5,19;13]_2})$, and  $[s_1\cdot \mathcal{P}_{[5]}\setminus (\mathcal{P}_{[4]} \setminus \mathcal{P}_{I_4})] $ generate $[31s_1+8,5,16s_1+3;13]_2$,  $[31s_1+12,5,16s_1+5;11]_2$, and $[31s_2+20,5,16s_2+9;8]_2$ AFER-optimal linear codes, where
			\begin{equation*}
				{\scriptsize
					\setlength{\arraycolsep}{0pt}
					G_{[39,5,19;13]_2}= \left( \begin{array}{*{39}{c}}
						1&0&0&0&0&0&1&1&1&1&1&1&0&0&0&0&0&0&1&1&1&1&1&1&0&0&0&0&0&0&1&1&1&0&1&1&1&0&1 \\
						0&1&0&0&1&1&1&0&0&0&1&1&0&1&1&1&0&0&0&1&1&1&0&0&1&0&0&0&1&1&1&0&0&1&0&1&1&0&1 \\
						0&0&1&0&1&1&1&0&0&1&0&0&0&1&1&0&1&1&1&0&0&1&0&0&0&1&1&0&1&1&1&0&0&1&1&0&0&1&0 \\
						0&0&0&1&1&1&1&1&1&0&0&0&0&0&0&1&1&1&1&1&1&0&0&0&0&0&0&1&1&1&1&1&1&0&0&0&0&1&1 \\
						0&0&0&0&0&0&0&0&0&0&0&0&1&1&1&1&1&1&1&1&1&1&1&1&1&1&1&1&1&1&1&1&1&0&1&1&1&0&0
					\end{array} \right),
					G_{[43,5,21;11]_2}= \left( \begin{array}{*{43}{c}}
						1&0&0&0&1&1&0&1&1&1&0&0&1&0&0&0&0&1&1&1&1&0&1&1&1&0&0&1&1&0&0&0&1&1&1&0&1&1&1&0&0&0&0 \\
						0&1&0&0&1&1&0&1&1&0&1&1&0&1&1&1&0&0&0&1&1&1&1&0&0&0&0&1&1&0&0&1&0&0&0&1&1&0&0&1&1&0&0 \\
						0&0&1&0&0&0&0&0&0&0&0&0&0&0&0&1&0&1&1&0&0&1&1&1&1&1&1&0&1&1&1&1&1&1&0&1&1&1&1&0&0&1&1 \\
						0&0&0&1&1&1&0&0&0&1&1&1&0&0&0&1&1&1&1&1&1&0&1&0&0&1&1&0&0&0&0&1&1&1&1&0&1&0&0&1&1&1&1 \\
						0&0&0&0&0&0&1&1&1&0&0&0&1&1&1&1&1&1&1&1&1&0&0&0&0&0&0&0&1&1&1&0&0&0&1&1&1&1&1&1&1&1&1
					\end{array} \right).}
			\end{equation*}

			Since those three linear codes all have odd minimum distances, their extended codes $\widehat{C}_1$, $\widehat{C}_2$, and $\widehat{C}_3$  have parameters $[31s_1+9,5,16s_1+4]_2$, $[31s_1+13,5,16s_1+6]_2$, and $[31s_1+21,5,16s_1+10]_2$. By Lemmas \ref{Tho_finite_geometry_results} and \ref{minihyper_kissing_number}, the first two codes are SS-type linear codes, and the third is SS or Belov-type linear codes.
            
            Let $\Im_i$, $1\le i\le 3$, be the associated minihyper of $\widehat{C}_i$.
            Then, it follows that
			$\Im_1=\mathcal{P}_{[4]}\bigcup  \mathcal{P}_{[3,5]}$ or $ \mathcal{P}_{[4]}\bigcup  \mathcal{P}_{[2,4]}$,
			$\Im_2=\mathcal{P}_{[4]}\bigcup  \mathcal{P}_{[4,5]}$ or $ \mathcal{P}_{[4]}\bigcup  \mathcal{P}_{[3,4]}$, and
			$\Im_3=  \mathcal{P}_{[3]}\bigcup  \mathcal{P}_{[4,5]}$, $\mathcal{P}_{[3]}\bigcup  \mathcal{P}_{[3,4]}$, $\mathcal{P}_{[3]}\bigcup  \mathcal{P}_{[2,3]}$ or $\mathcal{P}_{[4]}\setminus \mathcal{P}_{T_4}$. 
            It is easy to calculate the error coefficients of the punctured code $\widehat{C}_i^{\prime}$ of $\widehat{C}_i$. 
			After iterating over all feasible solutions, we obtain $e (\widehat{C}_1^{\prime})\ge 13$, $e (\widehat{C}_2^{\prime})\ge 11$, and $e (\widehat{C}_3^{\prime})\ge 8$. 
			Thus, $e_1=13$, $e_2=11$, and $e_3=8$. 
		\end{proof}

		Thus far, we determine the parameters of all binary $5$-dimensional AFER-optimal linear codes not listed in Table \ref{Five_Binary_AFER}. We present their parameters and specific constructions in Table \ref{Five_Binary_AFER_2}.
		
		Additionally, to evaluate the performance of our bounds when they are not tight, we include the lower bounds derived from Theorem \ref{Theo_Mix_lower_bound} in the fourth column of Table \ref{Five_Binary_AFER_2}. It is evident that the difference between our lower bounds and the actual values is less than or equal to $2$. This shows that our bounds perform well even when they are not tight.

		\begin{remark}\label{remark_datebase} 
            For binary Griesmer optimal linear codes with higher dimensions, we can use the AFER-optimal linear codes from Tables \ref{Three_Binary_AFER} to \ref{Five_Binary_AFER_2} as a foundation and iteratively apply Theorem \ref{Theo_Mix_lower_bound} to derive lower bounds on their error coefficients.
            The implementation can be broken down into the following steps:

            \textbf{Step 1:}  Create a parameter database $\mathcal{B}$ for binary AFER-optimal linear codes with dimensions $2$, $3$, and $5$.

            \textbf{Step 2:}  Using the available information (the database $\mathcal{B}$ and classification results on the maximum rank of Griesmer minihypers), apply Theorem \ref{Theo_Mix_lower_bound} to iteratively derive the error coefficient lower bounds for $6$-dimensional Griesmer optimal linear codes, and update the database $\mathcal{B}$ with the results. 

            \textbf{Step 3:}  Repeat step 2 in a similar way, expanding database $\mathcal{B}$ repeatedly. 
    

This method can help us build a database containing the error lower bounds of all Griesmer optimal linear codes. The same idea also applies to the non-binary case. 
The most immediate application is to help improve Grassl's codetable \cite{Grassl:codetables} and Maruta's codetable \cite{Maruta:codetable}.
		\end{remark}

		\begin{table*}[ht]
			\caption{ Five-dimensional AFER-optimal linear codes whose AFER-Optimality cannot be directly determined by Theorem \ref{Theo_Mix_lower_bound}.
			}\label{Five_Binary_AFER_2}
			\centering	
			\begin{tabular}{ccccc}
				\toprule
				Length   &          Parameters          &                            Constructions                             & \makecell[c]{Lower bounds on $A_d(C)$\\ from Theorems \ref{Theo_Mix_lower_bound} and \ref{L_extendability_Linear_codes}} &              AFER-Optimality               \\ \midrule
				
				11     &        $[11,5,4;4]_2$        &                          $G_{[11,5,4;4]_2}$                          &                                       $\ge 4$ (Theorem \ref{L_extendability_Linear_codes})             & \cite{ShitaoLi2024} $^\ast$ \\
				13     &        $[13,5,5;3]_2$        &                          $G_{[13,5,5;3]_2}$                          &                                        -                                         &            \cite{ShitaoLi2024} $^\ast$             \\
				$31s_2+7$  &  $[31s_2+7,5,16s_2+2;5]_2$   &            $([s_2\cdot \mathcal{P}_{[5]}],G_{[7,5,2;5]_2})$             &                                 $\ge 4$ (Theorem \ref{Theo_Mix_lower_bound}, Case 2)                                 &         Theorem \ref{Theo_[7,5,2]}         \\
				$31s_1+8$  &  $[31s_1+8,5,16s_1+3;13]_2$  &           $([s_2\cdot \mathcal{P}_{[5]}],G_{[39,5,19;13]_2})$           &                                $\ge 11$ (Theorem \ref{Theo_Mix_lower_bound}, Case 1)                                 &      Theorem \ref{Theorem_[21,5,10]}       \\
				$31s_1+9$  &  $[31s_1+9,5,16s_1+4;27]_2$  &   $(s_1+1)\cdot \mathcal{P}_{[5]}\setminus (\mathcal{P}_{[4]} \bigcup  \mathcal{P}_{[3,5]})$   &                                $\ge 25$ (Theorem \ref{Theo_Mix_lower_bound}, Case 3)                                 &         Theorem \ref{Theo_SS_h=2}          \\
				$31s_1+12$ & $[31s_1+12,5,16s_1+5;11]_2$  &           $([s_2\cdot \mathcal{P}_{[5]}],G_{[43,5,21;11]_2})$           &                                 $\ge 9$ (Theorem \ref{Theo_Mix_lower_bound}, Case 4)                                 &      Theorem \ref{Theorem_[21,5,10]}       \\
				$31s_1+13$ & $[31s_1+13,5,16s_1+6;23]_2$  &   $(s_1+1)\cdot \mathcal{P}_{[5]}\setminus (\mathcal{P}_{[4]} \bigcup  \mathcal{P}_{[4,5]})$   &                                $\ge 21$ (Theorem \ref{Theo_Mix_lower_bound}, Case 4)                                 &         Theorem \ref{Theo_SS_h=2}          \\
				$31s_2+20$ &  $[31s_2+20,5,16s_2+9;8]_2$  &   $s_1\cdot \mathcal{P}_{[5]}\setminus (\mathcal{P}_{[4]} \setminus \mathcal{P}_{I_4}) $    &                                 $\ge 7$ (Theorem \ref{Theo_Mix_lower_bound}, Case 1)                                 &      Theorem \ref{Theorem_[21,5,10]}       \\
				$31s_2+21$ & $[31s_2+21,5,16s_2+10;20]_2$ &   $s_1\cdot \mathcal{P}_{[5]}\setminus (\mathcal{P}_{[4]} \setminus \mathcal{P}_{T_4}) $    &                                $\ge 18$ (Theorem \ref{Theo_Mix_lower_bound}, Case 1)                                 &         Theorem \ref{Theo_SS_h=2}          \\
				$31s_2+23$ & $[31s_2+23,5,16s_2+11;14]_2$ &    $s_1\cdot \mathcal{P}_{[5]}\setminus (\mathcal{P}_{[3]}\bigcup  \mathcal{P}_{\{ 4\}} )$     &                                $\ge 13$ (Theorem \ref{Theo_Mix_lower_bound}, Case 1)                                 &         Theorem \ref{Theo_SS_h=2}          \\ \bottomrule
			\end{tabular}
			\begin{tablenotes}    
				\footnotesize               
				\item[1]	$^\ast $	These codes are not Griesmer optimal linear codes.
			\end{tablenotes}
		\end{table*}

		
		\section{Conclusion}\label{Sec VI}

		In this paper, we proposed five iterative lower bounds for the error coefficients of Griesmer optimal linear codes, as presented in Theorems \ref{Enhanced_GriesmerNound_q nmid}-\ref{Theo_minihyper_rank_bound} and Theorem \ref{Theo_Mix_lower_bound}. Numerical results demonstrated that our bounds are tight. Compared to the linear programming bound in \cite{sole2021linear}, our bounds are both computationally more efficient and tighter.
		Furthermore, we have fully determined the parameters of all binary AFER-optimal linear codes with dimensions up to $5$ and provided their explicit constructions, as shown in Tables \ref{Three_Binary_AFER}-\ref{Five_Binary_AFER_2}. 
        Finally, we also provided a method to establish lower bounds of the error coefficients of high-dimensional Griesmer optimal linear codes based on the results in this paper, as noted in Remark \ref{remark_datebase}. 

        Our work also raises some questions in this field. Here, we briefly discuss two questions directly related to our results. 
		
		\begin{itemize}
			\item 	
			Let $s$ be a non-negative integer and $C$ be an $[sv_k+n,k,sq^{k-1}+d;e_s]_q$ AFER-optimal linear code.
			Is there an integer $\imath $ such that $e_s$ is a constant integer for $s\ge \imath $?
			\item  Theorem \ref{Theo_minihyper_rank_bound} shows that the rank of associated Griesmer minihyper significantly affects the error coefficients of linear codes.
			A natural question is, can we establish bounds for minihyper to constrain its maximum rank?
			%
		\end{itemize}

			\section*{Appendix} \label{Appendix}

			To guarantee the integrity of this paper, we present generator matrices that are not listed in Tables \ref{Five_Binary_AFER} and \ref{Five_Binary_AFER_2} here.
			Specifically, it is easy to check that the following matrices generate $[7,5,2;5]_2$, $[8,5,2;1]_2$, $[9,5,3;4]_2$, $[10,5,4;10]_2$, $[11,5,4;4]_2$, $[12,5,4;1]_2$, $[13,5,5;3]_2$, $[14,5,6;7]_2$, $[39,5,19;13]_2$, $[42,5,20;3]_2$ and $[43,5,21;11]_2$ linear codes.

			\begin{equation*}
				{\scriptsize
					\setlength{\arraycolsep}{0pt}
					G_{[7,5,2;5]_2}=
					\setlength{\arraycolsep}{0pt}
					\left( {\begin{array}{*{10}{c}}
							1&0&0&0&0&1&0\\
							0&1&0&0&0&0&1\\
							0&0&1&0&0&1&1\\
							0&0&0&1&0&1&1\\
							0&0&0&0&1&1&1
					\end{array}} \right),
					G_{[8,5,2;1]_2}=
					\setlength{\arraycolsep}{0pt}
					\left( {\begin{array}{*{8}{c}}
							1& 0& 0& 0& 0& 1& 1& 1 \\
							0& 1& 0& 0& 0& 0& 1& 1 \\
							0& 0& 1& 0& 0& 1& 0& 1 \\
							0& 0& 0& 1& 0& 1& 1& 0 \\
							0& 0& 0& 0& 1& 1& 1& 1
					\end{array}} \right)
					,
					G_{[9,5,3;4]_2}=
					\left( {\begin{array}{*{9}{c}}
							1& 0& 0& 0& 0& 1& 1& 1& 1 \\
							0& 1& 0& 0& 0& 0& 1& 1& 1 \\
							0& 0& 1& 0& 0& 1& 0& 1& 1 \\
							0& 0& 0& 1& 0& 1& 1& 0& 1 \\
							0& 0& 0& 0& 1& 1& 1& 1& 0
					\end{array}} \right),
					G_{[10,5,4;10]_2}=
					\left( {\begin{array}{*{10}{c}}
							0&0& 0& 0& 1& 1& 1& 1& 0& 1\\
							1& 0& 0& 0& 0& 1& 1& 1& 1& 0 \\
							0& 1& 0& 0& 0& 0& 1& 1& 1& 1 \\
							0& 0& 1& 0& 1& 0& 1& 0& 1& 0 \\
							0& 0& 0& 1& 1& 0& 0& 1& 1& 0 \\
					\end{array}} \right),
				}
			\end{equation*}
			
			\begin{equation*}
				{\scriptsize
					\setlength{\arraycolsep}{0pt}
					G_{[11,5,4;4]_2}=
					\left( {\begin{array}{cccccccccccc}
							0&0& 0& 0& 1& 1& 1& 1& 0& 1& 1\\
							1& 0& 0& 0& 0& 1& 1& 1& 1& 0 &0 \\
							0& 1& 0& 0& 0& 0& 1& 1& 1& 1 & 0\\
							0& 0& 1& 0& 1& 0& 1& 0& 1& 0 & 0\\
							0& 0& 0& 1& 1& 0& 0& 1& 1& 0 & 0\\
					\end{array}} \right),
					G_{[12,5,4;1]_2}=
					\left( {\begin{array}{ccccccccccccc}
							0& 0& 0& 0& 1& 1& 1& 1& 0& 1& 1 &1\\
							1& 0& 0& 0& 0& 1& 1& 1& 1& 0& 0 &0\\
							0& 1& 0& 0& 1& 1& 0& 0& 1& 0& 1 &0\\
							0& 0& 1& 0& 0& 1& 0& 1& 1& 1& 1 &0\\
							0& 0& 0& 1& 1& 0& 0& 1& 1& 0& 0 &0
					\end{array}} \right),
					G_{[13,5,5;3]_2}=
					\left( {\begin{array}{ccccccccccccc}
							1 & 0 & 0 & 1 & 0 & 1 & 0 & 1 & 0 & 1 & 1 & 0 & 1 \\
							0 & 1 & 0 & 1 & 0 & 1 & 0 & 1 & 0 & 1 & 0 & 1 & 0 \\
							0 & 0 & 1 & 1 & 0 & 0 & 0 & 0 & 1 & 1 & 0 & 0 & 1 \\
							0 & 0 & 0 & 0 & 1 & 1 & 0 & 0 & 1 & 1 & 1 & 1 & 0 \\
							0 & 0 & 0 & 0 & 0 & 0 & 1 & 1 & 1 & 1 & 1 & 1 & 1
					\end{array}} \right),
				}
			\end{equation*}
			
			\begin{equation*}
				{\scriptsize
					\setlength{\arraycolsep}{0pt}
					G_{[14,5,6;7]_2}=
					\left( {\begin{array}{cccccccccccccc}
							1 & 0 & 0 & 0 & 0 & 1 & 1 & 0 & 0 & 1 & 1 & 1 & 1 & 0 \\
							0 & 1 & 0 & 0 & 1 & 0 & 1 & 0 & 1 & 0 & 1 & 1 & 0 & 1 \\
							0 & 0 & 1 & 0 & 1 & 1 & 0 & 0 & 1 & 1 & 0 & 1 & 0 & 0 \\
							0 & 0 & 0 & 1 & 1 & 1 & 1 & 0 & 0 & 0 & 0 & 1 & 1 & 1 \\
							0 & 0 & 0 & 0 & 0 & 0 & 0 & 1 & 1 & 1 & 1 & 1 & 1 & 1
					\end{array}} \right),
					G_{[42,5,20;3]_2}= \left( {\begin{array}{*{42}{c}}
							1& 0& 0& 0& 1& 1& 0& 1& 1& 1& 0& 0& 1& 0& 0& 0& 0& 1& 1& 1& 1& 0& 1& 1& 0& 0& 1& 1& 0& 0& 0& 1& 1& 1& 0& 1& 1& 1& 0& 0& 0& 0 \\
							0& 1& 0& 0& 1& 1& 0& 1& 1& 0& 1& 1& 0& 1& 1& 1& 0& 0& 0& 1& 1& 1&1& 0& 0& 0& 1& 1& 0& 0& 1& 0& 0& 0& 1& 1& 0& 0& 1& 1& 0& 0 \\
							0& 0& 1& 0& 0& 0& 0& 0& 0& 0& 0& 0& 0& 0& 0& 1& 0& 1& 1& 0& 0& 1&1& 1& 1& 1& 0& 1& 1& 1& 1& 1& 1& 0& 1& 1& 1& 1& 0& 0& 1& 1 \\
							0& 0& 0& 1& 1& 1& 0& 0& 0& 1& 1& 1& 0& 0& 0& 1& 1& 1& 1& 1& 1& 0&1& 0& 1& 1& 0& 0& 0& 0& 1& 1& 1& 1& 0& 1& 0& 0& 1& 1& 1& 1 \\
							0& 0& 0& 0& 0& 0& 1& 1& 1& 0& 0& 0& 1& 1& 1& 1& 1& 1& 1& 1& 1& 0&0& 0& 0& 0& 0& 1& 1& 1& 0& 0& 0& 1& 1& 1& 1& 1& 1& 1& 1& 1
					\end{array}} \right).
				}
			\end{equation*}

			
			\bibliographystyle{IEEEtran}
			\bibliography{reference}

\begin{thebibliography}{10}
\providecommand{\url}[1]{#1}
\csname url@samestyle\endcsname
\providecommand{\newblock}{\relax}
\providecommand{\bibinfo}[2]{#2}
\providecommand{\BIBentrySTDinterwordspacing}{\spaceskip=0pt\relax}
\providecommand{\BIBentryALTinterwordstretchfactor}{4}
\providecommand{\BIBentryALTinterwordspacing}{\spaceskip=\fontdimen2\font plus
\BIBentryALTinterwordstretchfactor\fontdimen3\font minus
  \fontdimen4\font\relax}
\providecommand{\BIBforeignlanguage}[2]{{%
\expandafter\ifx\csname l@#1\endcsname\relax
\typeout{** WARNING: IEEEtran.bst: No hyphenation pattern has been}%
\typeout{** loaded for the language `#1'. Using the pattern for}%
\typeout{** the default language instead.}%
\else
\language=\csname l@#1\endcsname
\fi
#2}}
\providecommand{\BIBdecl}{\relax}
\BIBdecl

\bibitem{Singleton1964}
R.~C. Singleton, ``Maximum distance $q$-ary codes,'' \emph{IEEE {T}ransactions
  on {I}nformation {T}heory}, vol.~10, p. 116–118, 1964.

\bibitem{griesmer1960bound}
J.~H. Griesmer, ``A bound for error-correcting codes,'' \emph{IBM Journal of
  Research and Development}, vol.~4, no.~5, pp. 532--542, 1960.

\bibitem{plotkin1960binary}
M.~Plotkin, ``Binary codes with specified minimum distance,'' \emph{IRE
  Transactions on Information Theory}, vol.~6, no.~4, pp. 445--450, 1960.

\bibitem{delsarte1972bounds}
P.~Delsarte, ``Bounds for unrestricted codes, by linear programming,''
  \emph{Philips Res. Rep.}, vol.~27, p. 272–289, 1972.

\bibitem{HengZiling2020}
Z.~Heng, Q.~Wang, and C.~Ding, ``Two families of optimal linear codes and their
  subfield codes,'' \emph{IEEE Transactions on Information Theory}, vol.~66,
  no.~11, pp. 6872--6883, 2020.

\bibitem{ShiMinjiaSO2023}
M.~Shi, S.~Li, and J.-L. Kim, ``Two conjectures on the largest minimum
  distances of binary self-orthogonal codes with dimension $5$,'' \emph{IEEE
  Transactions on Information Theory}, vol.~69, no.~7, pp. 4507--4512, 2023.

\bibitem{li2023hull}
Y.~Li, S.~Zhu, and E.~Mart{\'\i}nez-Moro, ``The hull of two classical
  propagation rules and their applications,'' \emph{IEEE Transactions on
  Information Theory}, vol.~69, no.~10, pp. 6500--6511, 2023.

\bibitem{HuZhao2024}
Z.~Hu, Y.~Xu, N.~Li, X.~Zeng, L.~Wang, and X.~Tang, ``New constructions of
  optimal linear codes from simplicial complexes,'' \emph{IEEE Transactions on
  Information Theory}, vol.~70, no.~3, pp. 1823--1835, 2024.

\bibitem{swaszek1995lower}
P.~F. Swaszek, ``A lower bound on the error probability for signals in white
  {G}aussian noise,'' \emph{IEEE {t}ransactions on {i}nformation {t}heory},
  vol.~41, no.~3, pp. 837--841, 1995.

\bibitem{abdullah2023some}
M.~Abdullah and W.~H. Mow, ``Some new constructions of {AFER}-optimal binary
  linear block codes,'' in \emph{2023 IEEE International Symposium on
  Information Theory (ISIT)}.\hskip 1em plus 0.5em minus 0.4em\relax IEEE,
  2023, pp. 1261--1265.

\bibitem{arikan2009channel}
E.~Arikan, ``Channel polarization: A method for constructing capacity-achieving
  codes for symmetric binary-input memoryless channels,'' \emph{IEEE
  {T}ransactions on {I}nformation {T}heory}, vol.~55, no.~7, pp. 3051--3073,
  2009.

\bibitem{arikan2019sequential}
E.~Ar{\i}kan, ``From sequential decoding to channel polarization and back
  again,'' \emph{arXiv preprint arXiv:1908.09594}, 2019.

\bibitem{polyanskiy2010channel}
Y.~Polyanskiy, H.~V. Poor, and S.~Verd{\'u}, ``Channel coding rate in the
  finite blocklength regime,'' \emph{IEEE Transactions on Information Theory},
  vol.~56, no.~5, pp. 2307--2359, 2010.

\bibitem{Rowshan2023}
M.~Rowshan and J.~Yuan, ``On the minimum weight codewords of {PAC} codes: The
  impact of pre-transformation,'' \emph{IEEE Journal on Selected Areas in
  Information Theory}, vol.~4, pp. 487--498, 2023.

\bibitem{Rowshan2023a}
M.~Rowshan, S.~H. Dau, and E.~Viterbo, ``On the formation of min-weight
  codewords of polar/{PAC} codes and its applications,'' \emph{IEEE
  Transactions on Information Theory}, vol.~69, no.~12, pp. 7627--7649, 2023.

\bibitem{Gu2024}
X.~Gu, M.~Rowshan, and J.~Yuan, ``Reverse {PAC} codes: Look-ahead list
  decoding,'' in \emph{2024 IEEE International Symposium on Information Theory
  (ISIT)}, 2024, pp. 2844--2849.

\bibitem{Dragoi2024}
V.-F. Drăgoi, M.~Rowshan, and J.~Yuan, ``On the closed-form weight enumeration
  of polar codes: $1.5d$-weight codewords,'' \emph{IEEE Transactions on
  Communications}, vol.~72, no.~10, pp. 5972--5987, 2024.

\bibitem{moradi2024polarization}
M.~Moradi, ``Polarization-adjusted convolutional ({PAC}) codes as a
  concatenation of inner cyclic and outer polar-and {R}eed-{M}uller-like
  codes,'' \emph{Finite Fields and Their Applications}, vol.~93, p. 102321,
  2024.

\bibitem{huffman2010fundamentals}
W.~C. Huffman and V.~Pless, \emph{Fundamentals of Error-Correcting
  Codes}.\hskip 1em plus 0.5em minus 0.4em\relax Cambridge university press,
  2003.

\bibitem{lin2020transformation}
C.-Y. Lin, Y.-C. Huang, S.-L. Shieh, and P.-N. Chen, ``Transformation of binary
  linear block codes to polar codes with dynamic frozen,'' \emph{IEEE Open
  Journal of the Communications Society}, vol.~1, pp. 333--341, 2020.

\bibitem{khebbou2023decoding}
D.~Khebbou, I.~Chana, and H.~Ben-Azza, ``Decoding of the extended {G}olay code
  by the simplified successive-cancellation list decoder adapted to
  multi-kernel polar codes,'' \emph{TELKOMNIKA (Telecommunication Computing
  Electronics and Control)}, vol.~21, no.~3, pp. 477--485, 2023.

\bibitem{khebbou2023single}
------, ``Single parity check node adapted to polar codes with dynamic frozen
  bit equivalent to binary linear block codes,'' \emph{Indonesian Journal of
  Electrical Engineering and Computer Science}, vol.~29, no.~2, pp. 816--824,
  2023.

\bibitem{macwilliams1977theory}
F.~J. MacWilliams and N.~J.~A. Sloane, \emph{The theory of error-correcting
  codes}.\hskip 1em plus 0.5em minus 0.4em\relax New York: Elsevier/North
  Holland, 1978.

\bibitem{sole2021linear}
P.~Sol{\'e}, Y.~Liu, W.~Cheng, S.~Guilley, and O.~Riou, ``Linear programming
  bounds on the kissing number of $q$-ary codes,'' in \emph{2021 IEEE
  Information Theory Workshop (ITW)}.\hskip 1em plus 0.5em minus 0.4em\relax
  IEEE, 2021, pp. 1--5.

\bibitem{abdullah2023new}
M.~Abdullah and W.~H. Mow, ``New search for the polarization-adjusted
  convolutional codes with respect to the {AFER}-optimality criterion,'' in
  \emph{2023 IEEE International Symposium on Information Theory (ISIT)}.\hskip
  1em plus 0.5em minus 0.4em\relax IEEE, 2023, pp. 1723--1728.

\bibitem{ShitaoLi2024}
S.~Li, G.~Luo, M.~Shi, and S.~Ling, ``On the error coefficients of asymptotic
  frame errorrate optimal binary linear codes,'' \emph{IEEE Transactions on
  Information Theory}, 2025, early access.

\bibitem{li2024characterization}
S.~Li and M.~Shi, ``Characterization and classification of binary linear codes
  with various hull dimensions from an improved mass formula,'' \emph{IEEE
  Transactions on Information Theory}, vol.~70, no.~5, pp. 3357--3372., 2024.

\bibitem{liu2023kissing}
Y.~Liu, W.~Cheng, O.~Rioul, S.~Guilley, and P.~Sol{\'e}, ``Kissing number of
  codes: A survey,'' \emph{Coding Theory and applications}, 2023.

\bibitem{helleseth1984further}
T.~Helleseth, ``Further classifications of codes meeting the {G}riesmer
  bound,'' \emph{IEEE {T}ransactions on {I}nformation {T}heory}, vol.~30,
  no.~2, pp. 395--403, 1984.

\bibitem{mceliece1991modifications}
R.~McEliece and G.~Solomon, ``Modifications of the {G}riesmer bound,''
  \emph{The Telecommunications and Data Acquisition Report}, 1991.

\bibitem{landjev2001geometric}
I.~Landjev, ``The geometric approach to linear codes,'' in \emph{Finite
  Geometries: Proceedings of the Fourth Isle of Thorns Conference}.\hskip 1em
  plus 0.5em minus 0.4em\relax Springer, 2001, pp. 247--256.

\bibitem{ward1999introduction}
H.~N. Ward, ``An introduction to divisible codes,'' \emph{Designs, Codes and
  Cryptography}, vol.~17, no. 1-3, pp. 73--79, 1999.

\bibitem{HAMADA1993229}
N.~Hamada, ``A characterization of some $[n,k,d;q]$-codes meeting the
  {G}riesmer bound using a minihyper in a finite projective geometry,''
  \emph{Discrete Mathematics}, vol. 116, no.~1, pp. 229--268, 1993.

\bibitem{solomon1965algebraically}
G.~Solomon and J.~J. Stiffler, ``Algebraically punctured cyclic codes,''
  \emph{{I}nformation and {C}ontrol}, vol.~8, no.~2, pp. 170--179, 1965.

\bibitem{belov1974construction}
B.~Belov, V.~Logachev, and V.~Sandimirov, ``Construction of a class of linear
  binary codes achieving the {V}arshamov-{G}riesmer bound,'' \emph{Problemy
  Peredachi Informatsii}, vol.~10, no.~3, pp. 36--44, 1974.

\bibitem{landjev2007weighted}
I.~Landjev and L.~Storme, ``A weighted version of a result of {H}amada on
  minihypers and on linear codes meeting the {G}riesmer bound,'' \emph{Designs,
  Codes and Cryptography}, vol.~45, no.~1, pp. 123--138, 2007.

\bibitem{Grassl:codetables}
M.~Grassl, ``{Bounds on the minimum distance of linear codes and quantum
  codes},'' Online available at \url{http://www.codetables.de}, 2007, accessed
  on 2025-1-7.

\bibitem{hamada1985characterization}
N.~Hamada, ``Characterization resp. nonexistence of certain {$q$}-ary linear
  codes attaining the {G}riesmer bound,'' \emph{Bull. Osaka Women's Univ.},
  vol.~22, pp. 1--47, 1985.

\bibitem{maruta2001extendability}
T.~Maruta, ``On the extendability of linear codes,'' \emph{Finite Fields and
  Their Applications}, vol.~7, no.~2, pp. 350--354, 2001.

\bibitem{bosma1997magma}
W.~Bosma, J.~Cannon, and C.~Playoust, ``The {M}agma algebra system {I}: The
  user language,'' \emph{Journal of Symbolic Computation}, vol.~24, no. 3-4,
  pp. 235--265, 1997.

\bibitem{Maruta:codetable}
T.~Maruta, ``{Bounds on $n_q(k,d)$ for linear codes of small dimensions},''
  Online available at \url{http://mars39.lomo.jp/opu/griesmer.htm}, 2022,
  accessed on 2025-1-7.

\end{thebibliography}

		\end{document}